\documentclass[12pt]{article}
\usepackage{amsmath}
\usepackage{graphicx}
\usepackage{enumerate}
\usepackage{natbib}
\usepackage{url} 
\usepackage{enumitem}
\usepackage{bbm}
\usepackage{algorithm}
\usepackage{algorithmic}
\usepackage{booktabs}
\usepackage{amssymb}
\usepackage{bm}
\usepackage{amsthm}

\newcommand{\blind}{1}

\newtheorem{theorem}{Theorem}
\newtheorem{res}{Result}

\newcommand{\E}[1]{\mathrm{E}\left[#1\right]}
\newcommand{\Var}[1]{\mathrm{Var}\left[#1\right]}

\newcommand{\imp}[1]{\mathcal{I}(\bm{#1})}
\newcommand{\impf}[2]{\mathcal{I}_{#1}(\bm{#2})}
\newcommand{\xstar}{\x^*}
\newcommand{\x}{\bm{x}}
\newcommand{\X}{\bm{X}}
\newcommand{\F}{\bm{F}}
\newcommand{\XA}{\X^A}
\newcommand{\FA}{\F^A}
\newcommand{\XU}{\X^U}
\newcommand{\FU}{\F^U}
\newcommand{\kpar}{\mathcal{K}_{par}}

\begin{document}

\def\spacingset#1{\renewcommand{\baselinestretch}%
{#1}\small\normalsize} \spacingset{1}

%%%%%%%%%%%%%%%%%%%%%%%%%%%%%%%%%%%%%%%%%%%%%%%%%%%%%%%%%%%%%%%

\if1\blind
{
  \title{\bf Feature calibration for computer models}
  \author{Wenzhe Xu \thanks{
    The authors gratefully acknowledge \textit{please remember to list all relevant funding sources in the unblinded version}}\hspace{.2cm}\\
     School of Science, Beijing University of Posts and Telecommunications, 
     \\Beijing 100876, China. wenzhexu@bupt.edu.cn\\
    Daniel B. Williamson \\
    Department of Mathematical Sciences, University of Exeter, UK \\ and The Alan Turing Institute, The British Library, London, UK\\
    Frederic Hourdin \\
    LMD-IPSL, Sorbonne University, CNRS, Paris, France\\
    and \\
    Romain Roehrig\\
    CNRM, Universit\'e de Toulouse, M\'et\'eo-France, CNRS, Toulouse, France
   }
  \maketitle
} \fi

\if0\blind
{
  \bigskip
  \bigskip
  \bigskip
  \begin{center}
    {\LARGE\bf Feature calibration for computer models}
\end{center}
  \medskip
} \fi

\bigskip
\begin{abstract}
Computer model calibration involves using partial and imperfect observations of the real world to learn which values of a model's input parameters lead to outputs that are consistent with real-world observations. When calibrating models with high-dimensional output (e.g. a spatial field), it is common to represent the output as a linear combination of a small set of basis vectors. Often, when trying to calibrate to such output, what is important to the credibility of the model is that key emergent physical phenomena are represented, even if not faithfully or in the right place. In these cases, comparison of model output and data in a linear subspace is inappropriate and will usually lead to poor model calibration. To overcome this, we present kernel-based history matching (KHM), generalising the meaning of the technique sufficiently to be able to project model outputs and observations into a higher-dimensional feature space, where patterns can be compared without their location necessarily being fixed. We develop the technical methodology, present an-expert driven kernel selection algorithm, and then apply the techniques to the calibration of boundary layer clouds for the French climate model IPSL-CM. 

\end{abstract}

\noindent%
{\it Keywords:}  
Kernel History Matching, Uncertainty Quantification, Gaussian Processes.
\vfill

\newpage
\spacingset{1.9} 

\section{Introduction}
Computer models are widely used to simulate the behaviour of complex systems. They can be used to simulate a complex system under different conditions, or into the future from present conditions, potentially to support decision making. Computer models have input parameters, some of which might control the type of simulation being done (such as the CO$_2$ atmospheric concentration in a climate simulation), whereas others are ``free" parameters that need to be calibrated before the model can be used for inference, prediction or decision support \citep{ kennedy2001bayesian, rougier2007probabilistic, hourdin2017art}.

Typically, computer simulators are expensive/time consuming to run. For example, a high-resolution coupled Earth system model may take several months to complete a single run spanning several hundreds of years \citep{rougier2009analyzing, williamson2017tuning}. To facilitate calibration of these models, Gaussian Process emulators can be trained on limited training sets of simulations and then used in place of the model in the calibration exercise \citep{asher2015review}. Though we make reference to fully Bayesian calibration \citep{kennedy2001bayesian} throughout, in this article we will focus on calibration via History Matching (HM). HM uses emulators to identify regions of input space that are likely to result in unacceptable mismatches between computer outputs and observations and iteratively removes those regions through a sequence of computer experiments. HM has been applied in many different disciplines, including oil reservoir modelling 
\citep{craig1997pressure,cumming2010bayes},
epidemiology \citep{andrianakis2015bayesian,andrianakis2017efficient} and
galaxy formation \citep{vernon2010galaxy,bower2010parameter}. It has gained particular popularity for climate models, where there is a strong appeal of the idea of retaining all models consistent with a set of important metrics to aid understanding and model development \citep{williamson2015identifying, Processbased1,Processbased2}.

When calibrating to high dimensional output such as spatio-temporal fields, emulators are either constructed ``cell by cell" with a separate GP for each scalar output \citep{gu2016parallel, salmanidou2021NatHaz}, or by projecting the output onto a basis and emulating its coefficients. In this article we focus on the latter, for a comparison of the approaches see \cite{salter2019efficient}. Calibration has been adapted for basis approaches, with the basis defined by the principal components of the training set being the most popular in practice \citep{higdon2008computer, wilkinson2010bayesian, chang2016calibrating}.

For complex spatio-temporal models, the number of training runs will be far fewer than the number of output dimensions, so any training set-derived basis will be rank deficient and unable to reproduce every spatio-temporal pattern that might be of interest. Specifically, basis methods do not necessarily preserve or represent what we will term `emergent features' of the modelled system, whose spatio-temporal location varies due to the parameters, and whose location in the real world may be different to their ideal location in the model. For example, consider the Gulf Stream as an emergent feature of an ocean model. It is known that the higher the resolution of a simulation, the narrower and stronger the current near the surface, and the closer the current to the east coast of North America. However, it is not feasible today to run global ocean model at high enough resolution to have a model Gulf Stream in exactly the same place as it is observed in observations. Yet certain parameters in a given ocean model affect the location of the Gulf Stream, and these require calibration. Approaches to calibrating these features generally focus on finding scalar metrics that characterise them \citep{Processbased1,Processbased2,williamson2012fast}. However, such metrics can be difficult to define and cannot fully capture a spatio-temporal process, so can be unsatisfactory to a model developer attempting to assess whether a particular parameterisation induces a reasonable representation of the emergent feature during a simulation. 

In this paper, we develop HM for emergent features by adapting kernel PCA: using a kernel to first project the output into a `feature space', and then emulating the coefficients of a PCA done in that space. As model output cannot be straightforwardly constructed from the emulator, we generalise the HM procedure to accommodate comparison of model output and data in the feature space. 
 In Section \ref{History matching for high-dimensional output}, we describe existing methods for calibration in many dimensions. We present kernel history matching in Section \ref{Kernel-based history matching}. Section \ref{sec:selection} introduces a method for user-supported kernel selection. In Section \ref{sec::application} we apply our method to the French climate model, IPSL-CM. Section \ref{Conclusion} contains discussion.

\section{History matching high-dimensional model output}
\label{History matching for high-dimensional output}

Let $f(\mathbf{x})$ denote a computer model mapping inputs  $\mathbf{x}\in \mathcal{X}$ into $\mathbb{R}^n$. Letting $z$ denote observations from the real world that the model attempts to mimic, History Matching (HM), like Bayesian calibration \citep{kennedy2001bayesian}, specifies a statistical model relating the computer model to the observations via a `best input' value, $\x^*$. The usual relationship is 
\begin{equation}
z=f(\x^*)+\eta+e,
\label{best}
\end{equation}
where $\eta$ is the model discrepancy, $e$ is the observation error, and all 3 terms are uncorrelated \citep{craig1996bayes, williamson2013history, vernon2010galaxy}.

The \cite{kennedy2001bayesian} calibration framework proceeds by specifying probability distributions for each term in (\ref{best}), with Gaussian processes modelling the function $f(\cdot)$ and discrepancy $\eta$. Priors on $\x^*$ and any hyperparameters of the joint specification are sufficient for Bayesian inference on the best input (and model discrepancy). Strongly informative priors for $\eta$ are required to overcome the identifiability problem raised both in the discussion to \cite{kennedy2001bayesian}, and presented in further detail by \cite{brynjarsdottir2014learning}.

History matching requires a weaker second order model specification with expectations and variances for all terms in (\ref{best}) (other than $\xstar$). The second order specification comes without probability, making expectation the primitive, and inducing a Hilbert space over the collection $\{z, f(\cdot), \eta, e\}$ with the covariance inner product \citep[see, for example][for more on expectation without probability]{williamson2013history, goldstein201326}. For any given value of $\x$, History matching exploits this geometry to calculate the distance between the observations and the model via 
\begin{equation}
  \mathcal{I}(\x)=(\textbf{z}- \E{f(\x)})^T\Var{\textbf{z}- \E{f(\x)}}^{-1}(\textbf{z}- \E{f(\x)}).
  \label{implausibility}
\end{equation}
where, under (\ref{best})
\begin{equation}
 \Var{\textbf{z}- \E{f(\x)}}=  \Var{f(\x)}+\Var{e}+\Var{\eta}.
 \label{varmat}
 \end{equation}
 $\mathcal{I}(\x)$ is termed the \textit{implausibility}, as large values indicate that  $f(\x)$ is too far from the observations for $\x$ to be $\xstar$. A user-specified cutoff threshold, $T$, defines what it means to be ``implausible", with all parameter values such that $\imp{x}>T$ ruled out, and the remaining parameter space termed `Not Ruled Out Yet' (NROY).
To select $T$, appeals have been made to Chebyshev's inequality, Pukelsheim's `three sigma rule' \citep{pukelsheim1994three} and comparison to the quantiles of a Chi-squared distribution with $r$ degrees of freedom (where $r$ is the rank of the variance matrix in (\ref{varmat}) \citep{vernon2010galaxy}). 
NROY space is identified through a sequence of HM exercises known as waves, performed within each previously derived NROY space \citep{Processbased1}. This can lead to all space being ruled out, a situation termed ``the terminal case" by \cite{salter2018uncertainty}, or to a subset of parameters that match the observations. 

\subsection{Emulation}
\label{Emulation}
The expectation and variance of $f()$ are normally obtained by building an emulator for the computer model. If the model is fast enough, it can be used directly with $\E{f(\x)} = f(\x)$, $\Var{f(\x)} = 0$ \citep{gladstone2012calibrated}. Any approach to emulation that gives a fast mean and variance can be used in History matching. For a comprehensive review of modern approaches, see \cite{gramacy2020surrogates}. Here we will outline the standard Gaussian Process (GP) emulator.
 
 Given a prior mean function $m(\x) = h^T(\x)\beta$, a correlation function, $c(\x, \x'; \delta)$, depending on correlation parameters $\delta$, variance parameter $\sigma^2$, and values for parameters $\beta, \delta$, the expectation and variance of $f(\x)$ having observed training set $\textbf{F}=(f(x_1) ,\ \dots  ,\  f(x_n))$ are

\begin{equation}
m^*(\mathbf{x}) = h^{T}(\mathbf{x})\beta + t(\mathbf{x})^T \Lambda^{-1} (\textbf{F}-\textbf{H}\beta),
\label{1eq:fornugget1}
\end{equation}
and $\sigma^2c^*(x,x')$, with $t(\mathbf{x})^T= (c(\mathbf{x},\mathbf{x}_1),\dots ,c(\mathbf{x},\mathbf{x}_n))$,
\begin{equation}
c^*(\mathbf{x},\mathbf{x}')=c(\mathbf{x},\mathbf{x}') - t(\mathbf{x})^T\Lambda^{-1} t(\mathbf{x}'),
\label{1eq:fornugget2}
\end{equation}
 and with $\Lambda$ an $n\times n$ covariance matrix with entries $\sigma^{2}c(\mathbf{x}_i,\mathbf{x}_j)$.

\subsection{Multivariate history matching with basis projection}
For a computer model with high-dimensional output, the most common method for calibration and emulation is to use dimension reduction techniques to represent the high dimensional output as linear combinations of a fixed set of low dimensional basis vectors \citep{higdon2008computer, wilkinson2010bayesian}. 

Suppose the model output, $f(\mathbf{x}_i)$, is a vector of length $l$, so that $\textbf{F}=(f(\mathbf{x}_1) ,\ \dots  ,\  f(\mathbf{x}_n))$ is a matrix that has dimension $l\times n$, and the $n$ inputs are $\mathbf{X}=(\mathbf{x}_1, \ldots, \mathbf{x}_n)$. Principal components can be obtained via singular value decomposition (SVD) of the centred output matrix $ \tilde{\textbf{F}}=\textbf{F}-\textbf{u}$, 
where $\textbf{u}$ is the ensemble mean $\textbf{u}=(u_1  ,\ \dots  ,\ u_l)^T$, and $u_i$ is the mean of $i$-th output, $ u_i=\frac{1}{n} \sum_{j=1}^n f_i(\mathbf{x}_j)$.
Denote the matrix whose columns contain the first $q$ principal components as $\Psi_q$. The number retained, $q$, is normally chosen so that the majority of the variance in the ensemble is explained by projection onto the basis. For instance, \cite{higdon2008computer} suggested that $\Psi_{q}$ should explain more than $99\%$ of the total variance of the data.

The vector of coefficients for the projection of the output onto basis $\Psi_q$ is 
\begin{equation*}
\textbf{c}(\mathbf{x})=(\Psi_{q}^T\Psi_{q})^{-1}\Psi_{q}^T(f(\mathbf{x})-\textbf{u}).
\end{equation*}
Emulators can then be fitted for each of the $q$ coefficients $c_i(\x)$ individually, with the expectation and variance of $\textbf{c}(\x)$ defined as $\E{\textbf{c}(\x)}=(\E{c_1(\x)}, \ldots, \E{c_q(\x)})^T,$
and the associated emulator variance matrix
$\Var{\textbf{c}(\x)}= \mathrm{diag}(\Var{c_1(\x)}, \ldots, \Var{c_q(\x)})$.

Multivariate history matching with basis emulators can either be applied in the original output space via (\ref{implausibility}) and $$\E{f(\mathbf{x})}= \Psi_{q} \E{c(\mathbf{x})}, \qquad \Var{f(\mathbf{x})}= \Psi_{q}\Var{c(\mathbf{x})}\Psi_{q}^T,$$ or on the coefficient space with, 
\begin{equation}
    \mathcal{I}_c(\x)=(\textbf{c}(\textbf{z})- \E{c(\x)})^T \left(  \Var{\textbf{c}(\eta)}+  \Var{\textbf{c}(e)}+\Var{c(\x)}\right)^{-1}(\textbf{c}(\textbf{z})- \E{c(\x)}),\\
    \label{1eq:coefficient muti-Implausibility}
\end{equation}
where $\textbf{c}(\textbf{z}) = (\Psi_{q}^T\Psi_{q})^{-1}\Psi_{q}^T(\mathbf{z}-\textbf{u})$, $ \Var{\textbf{c}(e)}=\Var{\Psi_q^T(e-\textbf{u})} =\Psi_q^T\Var{e}\Psi_q$, and $  \Var{\textbf{c}(\eta)}=\Var{\Psi_q^T(\eta-\textbf{u})} =\Psi_q^T\Var{\eta}\Psi_q$. Computationally, the coefficient implausibility is often the only feasible tool with which to perform history matching, yet used directly it does not lead to the same space reduction as the full field implausibility. \cite{salter2019efficient} showed that the two implausibilities differ by a constant which can be found via a single computation with the same complexity as (\ref{implausibility}).

\cite{salter2018uncertainty} discuss how history matching with the above PCA basis methods can cause a false terminal case when the observations are not in the linear subspace spanned by $\Psi_q$. A false terminal case occurs when there exist values of $\x$ consistent with (\ref{best}) but when history matching rules out all of the parameter space.
\cite{salter2018uncertainty} proposed a rotation that overcomes this problem. The technique amounts to finding an orthogonal basis that minimises the error induced by projecting the observations themselves onto the basis and then reconstructing them from the resulting coefficients, subject to a constraint on the total variability explained by projection onto the basis vectors. 

In many practical applications, what is important is that the model has the right emergent features or ``patterns" given its complexity, and not that those features directly reproduce observations. Such features are particularly common in climate modelling, with examples including position and strength of atmospheric jets, rain bands with respect to topography, strength of oceanic currents with respect to coastline and bathymetry and the El Ni\~no Southern Oscillation (ENSO) spatial pattern and its associated teleconnections. A good simulation must have these properties and it is usually the case that the resolution of the solver would preclude them from being in the same location (gridcells/time points) as observations. For these problems, the distance defined in (\ref{implausibility}) is inappropriate. 

Consider, as a simple example, an observation of a horizontal line of 1's on the bottom 10 cells of a 10x10 grid, with 0's elsewhere (e.g. to represent an ocean current at the discretisation of a toy ocean model). Under any discrepancy with finite variance, the simulation with 0's everywhere, is closer to our observed line of 1's on the bottom 10 cells, than any horizontal line of 1's on the remaining 9 rows. In short, a representation of the current in the wrong place is worse than no current at all under (\ref{implausibility}).

In the next section we adapt kernel methods for history matching. The geometry of history matching, as shown in Section \ref{Emulation}, is fundamentally linear (given that the inner product is covariance). Kernel methods relax that linearity by first mapping the original output space to a higher-dimensional feature space and allowing the linearity only to apply there. We first overview the key ideas from kernel methods before adapting history matching accordingly.

\section{Kernel-based history matching}
\label{Kernel-based history matching}
In this section we describe Kernel Principal Component Analysis (KPCA) for the emulation of computer models \citep{scholkopf2002learning}. Whilst KPCA emulation has appeared in the literature \citep{xing2016manifold}, we describe it here in detail to facilitate our adaptation of history matching.

\subsection{Basis emulation with Kernel PCA}\label{kpca}

Consider model output $f(\x)$ living in $l$-dimensional output space. Kernel methods map $f(x)$ into an $m \geq l$ dimensional ``feature" space via a typically unknown mapping function $\phi(f(\x))$. Rather than specifying $m$ or $\phi$, this mapping is characterised by a user-specified positive definite kernel, $k(\cdot, \cdot)$, that operates on vectors in our output space and returns the dot product of the mapped vectors in feature space \citep{bishop2006pattern,hofmann2008kernel,genton2001classes, scholkopf2001learning}. I.e. 
\begin{equation}
    k(f(\x),f(\x')) = <\phi(f(\x)),\phi(f(\x'))>.
    \label{C4:kerneldefination1}
\end{equation}

Even without access to the explicit mapping, $\phi(\cdot)$, access to the dot product in feature space means we can perform certain analyses within the feature space as if we were able to observe the mapped vectors directly. This works whenever an algorithm using the mapped vectors $\phi(f(\x))$ and $\phi(f(\x'))$ can be expressed only in terms of $\phi(f(\x))^T\phi(f(\x'))$, because the inner products can be replaced by $k(f(\x), f(\x'))$ (the `kernel-trick'), giving access to explicit calculations needed for inference. PCA is one such algorithm, leading to kernel PCA (KPCA) a feature space analogy to the PCA basis methods described in Section \ref{History matching for high-dimensional output}.

Consider the (theoretical) $m \times n$ matrix representing the centred mapped ensemble,
$\tilde{\Phi}=(\tilde{\phi}(f(\x_1)), \ldots, \tilde{\phi}(f(\x_n)))$, such that $\tilde{\phi}(f(\x_n))=\phi(f(\x_i))-\bar{\phi}$ with  mapped ensemble mean $\bar{\phi}=(\bar{\phi}_1  ,\ \dots  ,\ \bar{\phi}_m)^T$, $ \bar{\phi}_i=\frac{1}{n} \sum_{j=1}^n \phi(f_i(\mathbf{x}_j))$.
The kernel PCA basis is $\Psi=(\psi_1  ,\ \dots  ,\ \psi_{n})$, where each $\psi_k$ has length $m$ and the $k$th basis vector $\psi_k$ is 
\begin{equation}
  \psi_k=\sum_{i=1}^{n}\alpha_{ki}\tilde{\phi}(f(\x_i)) = \bm{\alpha}_k\tilde{\Phi}. 
    \label{wk}
\end{equation}
Here $\bm{\alpha}_k$ is the $k$th eigenvector of the centred kernel matrix, $\tilde{\bm{K}} = <\tilde{\phi}(f(\x)),\tilde{\phi}(f(\x'))>$, normalised in feature space \citep[full proof of this result is given in][]{scholkopf1998nonlinear}. The centred kernel matrix can be expressed in terms of the corresponding kernel matrix, $\bm{K}$ (with $ij$th entry $k(f(\x_i), f(\x_j))$) via $$\tilde{K}_{ij} = K_{ij} - \frac1n\sum_{l=1}^nK_{il} - \frac1n\sum_{l=1}^nK_{lj} + \frac{1}{n^2}\sum_{k=1}^n\sum_{l=1}^nK_{kl},$$ and the normalised eigenvectors of $\tilde{\bm{K}}$ are equivalent to its standard eigenvectors, scaled by the square root of the eigenvalues \citep{scholkopf1998nonlinear}.

Though the coefficients in (\ref{wk}) are obtainable, the basis vectors themselves involve the generally unavailable mapped ensemble. However, the projection from feature space into a low dimensional coefficient space is available through the kernel trick via

\begin{equation}
 C_k(\x)= \psi_k^T\tilde{\phi}(f(\x))
 =\sum_{j=1}^{n}\alpha_{kj}\tilde{\phi}(f(\x_j))^T\tilde{\phi}(f(\x))
 =\sum_{j=1}^{n}\alpha_{kj}\tilde{k}(f(\x_j),\ f(\x)).
  \label{projection}
\end{equation}
Given that the above describes a PCA in feature space, the usual PCA properties, such as variance explained by each component in order is maximised subject to orthogonality with the previous components, hold so that we can truncate the KPCA basis to the first $q$ basis vectors, and build emulators for the coefficients $C_1(\x), \ldots, C_q(\x)$ as described in Section \ref{History matching for high-dimensional output} for standard PCA. 

\subsection{History Matching in feature space}\label{sec::KHM}

By defining feature space as the natural space in which to compare model outputs to each other and to observations, the goal of history matching is to find runs with small distances to the observations in feature space. I.e. we want
\begin{equation}
 \begin{aligned}
 \mathcal{I}_{0}(\x)&= ||\phi(z) - \phi(f(\x))||^2 \\
 &= \left(\phi(z)-\phi(f(\x))\right)^T
 \left(\phi(z)-\phi(f(\x))\right)\\
&=k(z,z)+k(f(\x),f(\x))-2k(f(\x),z),\\
 \end{aligned}
 \label{if0}
\end{equation}
the distance between a model output and the observations in feature space, to be small. As with standard history matching, we must adapt (\ref{if0}) to account for only having an emulator for $\phi(f(\x))$ and we do this in the subsequent subsections. Here we discuss the treatment and meaning of the observation and structural uncertainties represented in (\ref{best}).

If we adopt (\ref{best}) on the output space, rather than feature space, we do not obtain a feature representation of (\ref{best}). I.e. $$\phi(z) \neq \phi(f(\xstar)) + \phi(\eta) + \phi(e).$$ This means that, although it is possible to project the uncertainty on $\eta$ and $e$ into feature space and then down onto coefficient space to compute a feature-space analogy to (\ref{1eq:coefficient muti-Implausibility}) \citep[see][Section 4.4 for details]{xu2021generalising}, performing a history match in this way targets the incorrect space (if we hold to (\ref{best})). Instead, we appeal to the fundamental geometric idea behind History Matching: namely, rule out models that are `far' from observations when accounting for uncertainty. Instead of beginning with (\ref{best}), then quantifying key uncertainties and then cutting out space, we begin with a kernel $k(\cdot, \cdot)$ that builds these notions of uncertainty into the measure of how close two points are in the feature space. Given such a kernel we need only adapt (\ref{if0}) to account for emulator uncertainty and find a suitable implausibility cutoff in order to history match in feature space (as we cannot appeal to (\ref{best})). 

One way to build such a kernel is to include observation error and any structural error judgements in the kernel via an $l \times l$ weight matrix, $W$, to reflect judgements regarding what are key regions of output space for matching observations and which are less important. For example, the Gaussian kernel with parameters $\sigma, \delta$ and weight matrix $W$ is
\begin{equation}
k(f(\x),f(\x'))=
\sigma\exp(-(f(\x)-f(\x'))^TW^{-1}(f(\x)-f(\x')))/\delta).
\label{nonlinear kernel}
\end{equation}
For another example, the linear kernel with weight $W$ is
\begin{equation}
k(f(\x),f(\x^{'}))= f(\x)^TW^{-1}f(\x^{'}).
\label{linear kernel}
\end{equation}
We view $W$ as a natural way to reflect judgements regarding which are key regions of output space for matching observations and which are less important. If observation error is available as a known matrix, $\Upsilon$, say, perhaps $W = \Upsilon + \tilde{W}$, where $\tilde{W}$ is set to reflect these discrepancy-type judgements. We present a kernel selection algorithm using expert judgement in Section \ref{sec:selection}. 

\subsection{Implausibility in feature space}
\label{Implausibility in feature space}
We now present 2 approaches to setting implausibility and appropriate cut off thresholds for history matching in feature space using an emulator. 

\subsubsection{Implausibility $\impf{1}{\x}$}

This first approach substitutes the emulator expectation, $\E{\phi(f(\x))}$, for $f(\x)$ in (\ref{if0}), as with the implausibility for standard history matching, and then uses the emulator uncertainty to develop a variable threshold. 

Define implausibility, $\impf{1}{\x}$, as 
\begin{equation}
\impf{1}{\x} = (\phi(z)-\E{\phi(f(\x))})^T
(\phi(z)-\E{\phi(f(\x))}).
\label{if1}
\end{equation}
\begin{theorem}
Consider a KPCA emulator for $f(\x)$ truncated to the first $q$ basis vectors so that $\E{\phi(f(\x))} = \Psi_q\E{\textbf{C}_q(\x)}+\bar{\phi}$, with $\E{\textbf{C}_q(\x)}$ the emulator fitted to the first $q$ KPCA coefficients. In this case, 
\begin{equation}
\impf{1}{\x}
=\tilde{k}(z,z)+\E{\textbf{C}_q(\x)}^T\E{\textbf{C}_q(\x)}-2\E{\textbf{C}_q(\x)}^T\textbf{A}\tilde{\textbf{K}}_{z},
\end{equation}
where $\textbf{A}$ is the matrix containing the first $q$ eigenvectors of the centred kernel matrix, $\tilde{\textbf{K}}_{z}=[\tilde{k}(z,f(\x_1)),\tilde{k}(z,f(\x_2)) ,\ \dots  ,\ \tilde{k}(z,f(\x_n))]$.
\end{theorem}
The proof is in S1.1.

\begin{res}\label{Tx}
Suppose we have fitted Gaussian process emulators for each $C_k(\x), k=1, \ldots, q.$ If the model and data are sufficiently close in feature space, then $$P\left(\impf{1}{\x} > T(\x)\right) < 0.1,$$  where 
\begin{equation}
T(\x)=\sum_{k=1}^q \Var{C_k(\x)}+ 3 \sqrt{2\left(\sum_{k=1}^q \Var{C_k(\x)}\right)^2}+a,
\label{threshold for if1}
\end{equation}
where constant $a$ is an upper bound on $||\phi(z) - \phi_q(f(\xstar))||^2$ that can be set by expert judgement or derived from the ensemble. 
\end{res}
The proof of this result is in S1.2 and the bound on the probability of exceedance can be made arbitrarily small by increasing the multiplier currently set to $3$. In practice, because the bound on the probability actually applies to an upper bound derived for $\impf{1}{\xstar}$, we have found setting $T(\x)$ with the exceedance bound of $0.1$ to be sufficient for targeting NROY space. We present setting or deriving $a$ in Section \ref{sec:selection}.

Implausibility $\impf{1}{\x}$ together with Result \ref{Tx} offer an algorithm for history matching in feature space with parameter-dependent cut-off threshold $T(\x)$. Namely, NROY space for a particular wave is defined to be $$\{\x : \impf{1}{\x} \leq T(\x)\}.$$

\subsubsection{Implausibility $\impf{2}{\x}$}
\label{sec:if2}
The above version of KHM places observational and structural uncertainty judgements into the kernel and uses euclidean distance in feature space to rule out points with emulator uncertainty appearing in the bound for implausible distances. An approach that may seem closer to standard history matching is to include the emulator uncertainty within the implausibility and to choose a constant bound. 
Let implausibility
\begin{equation}
\impf{2}{\x}= \big(\phi(z)-\E{\phi(f(\x))}\big)^T\big(\textbf{1}_m+\Var{\phi(f(\x))}\big)^{-1}
\big(\phi(z)-\E{\phi(f(\x))}\big),
\label{if2}
\end{equation}
the distance between the observation and model output in feature space scaled by the emulator uncertainty, where $\textbf{1}_m$ is the $m \times m$ identity matrix. Note that
\begin{equation*}
\impf{2}{\x}\longrightarrow \impf{1}{\x}  \qquad \text{as } \qquad
\Var{\phi(f(\x))} \longrightarrow 0.
\end{equation*}

\begin{theorem}
\begin{equation}
\impf{2}{\x} = \big(\textbf{C}_q(z)-\E{\textbf{C}_q(\x)}\big)^T\big(\Var{\textbf{C}_q(\x)}+ \textbf{1}_q\big)^{-1}\big(\textbf{C}_q(z)-\E{\textbf{C}_q(\x)}\big) + ||\varepsilon_z||^2,
\label{if2 coefficient way}
\end{equation}
where $\varepsilon_z$ is the observation reconstruction error and \begin{equation}
\begin{aligned}
 ||\varepsilon_z||^2&=(\tilde{\phi}(z)-\tilde{\phi}_q(z))^T(\tilde{\phi}(z)-\tilde{\phi}_q(z))\\
&=\tilde{k}(z,z)+\textbf{C}_q(z)^T\textbf{C}_q(z)-2\textbf{C}_q(z)^T\textbf{A}\tilde{\textbf{K}}_{z},
   \end{aligned}
   \label{varepsilon calculate}
\end{equation}
where $\textbf{C}(z)$ is the projection of $\phi(z)$ onto the coefficient space (found by substituting $z$ for $f(\x)$ in (\ref{projection})).
\end{theorem}
The proof is given in S1.3. As all uncertainties are accounted for either by the kernel or within the implausibility, a constant threshold $T_2$ can be applied as cut-off for ruling out space with $\impf{2}{\x}.$ This can be user specified, and we give a suggestion in the next section. 
 
 \section{Selecting kernels and bounds}\label{sec:selection}

The efficacy of any calibration method depends on user-input. From prior modelling of the input space, the assessment of uncertainties and the choice of outputs/emulators to calibrate with. For KHM, the key user input is the selection of the kernel. The kernel function and its parameters determine the characteristics of the feature space and define what it means for model output and data to be `similar'. 
Non-linear kernels may be the superior choice when history matching emergent features such as those described so far. However, they may not be suitable for every application, particularly when there are no non-linear structures within the data. For each new application, predetermining whether to use a linear kernel, or a non-linear kernel would be difficult. We propose a mixture kernel \citep{smits2002improved} where, for a given setting of the mixing parameters, standard history matching on the model output space is achieved. Our kernel is
\begin{equation}
k(f(\x),f(\x’))=\omega k_1(f(\x),f(\x’))+(1-\omega)k_2(f(\x),f(\x’)),
\label{mixture kernel}
\end{equation}
where $\omega \in [0,1]$ is a weight parameter, $ k_1(f(\x),f(\x’))$ is the linear kernel defined in (\ref{linear kernel}), and $k_2(f(\x),f(\x’)$ is a non-linear kernel. One possible choice for $k_2$ is the Gaussian kernel given in equation (\ref{nonlinear kernel}), and we use this choice in our application. 

\begin{theorem}
\label{standardHM}
Consider Kernel history matching using $\impf{2}{\x}$, the mixture kernel given in (\ref{mixture kernel}), and $W$ in (\ref{linear kernel}) as discussed in Section-\ref{sec::KHM}. Setting $\omega=1$ recovers standard History matching with the PCA basis and implausibility (\ref{implausibility}).
\end{theorem}
A proof is given in S1.4. Note that this result implies that Kernel History Matching generalises history matching and that the simulator output space is only one possible choice of feature space in which we can compare models and data.

\subsection{User Input}\label{sec::acceptable}
Even given the mixture kernel (\ref{mixture kernel}), there is a considerable burden on the user to select $k_2$, set the parameters of the kernels (including weight matrix $W$, weight $\omega$, and any non-linear kernel specific parameters), and to set the relevant components of the threshold of the implausibility (either $a$ for $\impf{1}{\x}$ or $T_2$ for $\impf{2}{\x}$). These choices define what it means for model outputs and data to be `close' in feature space and thus define what features themselves are. As such, user input is critical, yet it is unlikely that any user would be able to provide this input through direct specification of the required parameters. 

Instead, as with more standard types of prior elicitation, we propose to allow the user to specify judgements in terms of observables and to derive appropriate choices of parameters and kernels from there. Specifically, we work with the ensemble $(\X, \F)$, or a subset of it, and ask a user to divide it into a set of runs that are `close' to the data, $(\XA, \FA)$, and a set of `unacceptable' runs $(\XU, \FU)$. Before showing how we can use this information to provide implausibility bounds and to select kernels, we should say more on the feasibility and interpretation of this classification of (a subset of) the ensemble members.

The first important thing to discuss is what we mean by `close'. In the case where a user looks at the ensemble and immediately sees runs with acceptable representations of the process, the task of classifying the ensemble is simple enough. When there are no runs with the desired representation (or perhaps only one or two), we ask the users to put runs into $\XA$ that have at least some representation of the process or, if those don't exist either, at least look closer to it than the runs left in $\XU$. The classification should be repeated in each wave, enabling better simulations to evolve the meaning of ``close" simulations as the exercise proceeds.

\subsection{Kernel estimation}\label{sec::kernel.estimation}

Let $\kpar$ denote the unknown parameters in the mixture kernel: weight  $\omega \in [0,1]$, kernel parameters for $k_2$, $\kappa$, and any parameters, $\phi_{W}$, used to specify $W$. Define `true' NROY space under $\kpar$ and threshold $T$ via 
\begin{equation}
 \mathcal{X}_{\mathrm{NROY}}(\kpar, T)=\{\x \in\mathcal{X} :\impf{0}{\x} \leq T\}.
 \label{optimisation nroy2}
\end{equation}

Our aim is to define $\kpar, T$ so that all $\x^A\in\XA$ and all $\x^U\in\XU$ when possible. We first define a cost function that trades off accuracy and efficiency for a particular classification of the training data. Let
\begin{displaymath}
N_A(\kpar, T) = \sum_{\x\in\XA}\mathbb{I}(\x\in\mathcal{X}_{\mathrm{NROY}}(\kpar, T)), \quad N_U(\kpar, T) = \sum_{\x\in\XU}\mathbb{I}(\x\in\mathcal{X}_{\mathrm{NROY}}(\kpar, T)),
\end{displaymath}
and suppose the number of acceptable runs is $n_A$. Then our cost function, 

\begin{equation}
\mathcal{P}(\kpar, T)= \frac{\alpha}{n_A}  N_A(\kpar, T) +(1-\alpha)\frac{N_A(\kpar, T)}{N_A(\kpar, T)+ N_U(\kpar, T)},
\label{c5:performance evaluation}
\end{equation}
trades off accuracy, here the proportion of NROY acceptable runs, with efficiency, the proportion of NROY runs that are correctly classified. Here $\alpha\in[0,1]$ is set by the user to trade these terms off and we have found $\alpha=0.8$ to be a good practical choice. Note that the optimal value, $\mathcal{P}(\kpar, T)=1$, occurs when $$\max_{x\in\XA}(\impf{0}{\x}) \leq T \leq \min_{x\in\XU}(\impf{0}{\x}),$$ and so, even if an optimal $\kpar$ and $T$ exists, the choice of $T$ is not unique. We define the threshold in (\ref{optimisation nroy2}), $T^{**}$, to be the value of $T$ that keeps the largest NROY space subject to maximising (\ref{c5:performance evaluation}), i.e.
\begin{equation}\label{tstar}
    T^{**} = \max\left\{T^* : T^* = \arg\max_T\mathcal{P}(\kpar, T) \right\}.
\end{equation}
We then define our kernel parameters by maximising $\mathcal{P}(\kpar, T^{**})$.

\subsection{Bounds for $\impf{1}{\x}$ and $\impf{2}{\x}$}\label{sec::user-bounds}

As the value of $a$ in equation (\ref{threshold for if1}) represents an upper bound for $||\phi(z) - \phi_q(f(\xstar))||^2$, we can set this immediately to $T^{**}$ so that, for $\impf{1}{\x}$, the parameter dependent implausibility threshold becomes $$T(\x)=\sum_{k=1}^q \Var{C_k(\x)}+ 3 \sqrt{2\left(\sum_{k=1}^q \Var{C_k(\x)}\right)^2}+T^{**}.$$

An appropriate cut off threshold for $\impf{2}{\x}$ in (\ref{if2}) can be derived from the user-classified training set by noticing that, if we were already to have a KPCA coefficient emulator but had not used the set $\XU$ as part of its training set, then $T_2 = \min_{\x\in\XU}(\impf{2}{\x})$ is an upper bound on a useful threshold. Typically, we will have used all of the ensemble data to train an emulator, so we propose to derive $T_2$ by calculating $\impf{2}{\x}$ using a leave one out emulator for each member of $\XU$ (whereby the emulator is trained using all ensemble members except the left out member) and setting $T_2$ to be the smallest of these implausibilities.

\section{Process-tuning for climate models}
\label{sec::application}

Climate models solve the Navier–Stokes equations on a rotating sphere to simulate the evolution of the Earth's climate \citep{gettelman2016demystifying}. Solving the equations requires the introduction of a grid on which they can be discretized and then numerically integrated. This part of climate models (especially for their atmosphere and ocean compoments) is known as `the dynamical core'. The spatial resolution of the chosen grid determines the types of processes that can be resolved explicitly, and those that require a specific treatment (because they feed back onto the resolved state). The latter processes include turbulent mixing, convection transport and clouds, together with those that rely on physical equations not included in Navier-Stockes (e.g., radiation). These processes are represented through  `physical parametrisations', each of which introduces several `free' parameters that need to be calibrated. Typically climate model calibration targets appropriate behaviour at the process scale as well as good representation of the mean climate and its variability at various scales \citep{hourdin2017art}.

Recently there has been a growing interest in using Uncertainty Quantification to assist with climate model tuning \citep{hourdin2017art}. History matching has emerged as a popular candidate method, as pure optimisation methods and even Bayesian calibration can over fit the parameters controlling massive state vectors when using relatively sparse and selective observations \citep{salter2018uncertainty}. HM has been used to find good high resolution simulators for the UK climate projections \citep{sexton2021perturbed}, for ocean only models \citep{williamson2017tuning}, land surface models \citep{mcneall2016impact,baker2022emulation}, aerosol models \citep{johnson2020robust} and land ice models \citep{mcneall2013potential,gandy2022tuning}. \cite{Processbased1} demonstrated how HM can be used to calibrate physical parametrisation at the process level using single column versions of an atmospheric model compared to reference large-eddy simulations (LES) resolving the targeted processes explicitly. \cite{Processbased2} further showed that this process-based calibration and the iterative refocussing framework can be used to efficiently precondition the full atmospheric model calibration and thereby reduce the computational burden of full model simulations to achieve it. Other applications of HM for process-based calibration include \cite{audouin2021modeling} and \cite{Processbased3}. These process-based history matching examples use user-defined scalar metrics to capture the main features of boundary-layer clouds. PCA-based approaches fail because the targeted key clouds features that the model should capture emerge at different times, and can be slightly shifted vertically due to the coarse vertical resolution. As such shifts in the cloud field do not necessarily preclude of the quality of the simulation, this problem inspired the developments in this paper.

We apply KHM to the boundary-layer clouds simulated by the IPSL climate model, following the process-based framework developed by \cite{Processbased1}. The single-column version of the IPSL climate model is used to calibrate several parameters involved in the model convection parametrisation, based on a case study featuring the transition from stratocumulus to cumulus clouds that is ubiquitous over the subtropical oceans \citep[Reference case in][]{deRoode2016}. The single-column simulations are compared to a Large Eddy Simulation (LES) using the same initial and boundary conditions. The LES and model output represent hourly evolution of the cloud fraction profiles between the surface and 3 km above it, over the three days of the case study. The LES is shown in Figure \ref{fig:6.1}. When running the single-column model on this case study for various choices of its parameters, large differences of this cloud fraction time evolution can occur (see Figure given in S2.1). Nevertheless, important features of the LES that we would like the single-column model to capture include:
\begin{itemize}[itemsep=0cm]
    \item the overall increase of the cloud top over the 3 days, at the right speed
    \item the deepening of the cloud layer, with weak cloud fraction below a cloud layer reaching a cloud fraction around 1
    \item the diurnal cycle of the cloud fraction, with reduced cloud thickness during daytime
    \item the reduction of cloud fraction near the cloud top at the end of the simulation (breaking of stratocumulus clouds by more active convection).
\end{itemize}
In contrast, exact timing, cloud fractions, cloud bottom and top, or depths of the cloud layers are not critical to calibration, both because of structural limits of the model and the chaotic nature of the atmosphere.

\begin{figure}
\centering
\includegraphics[height=0.35\textheight]{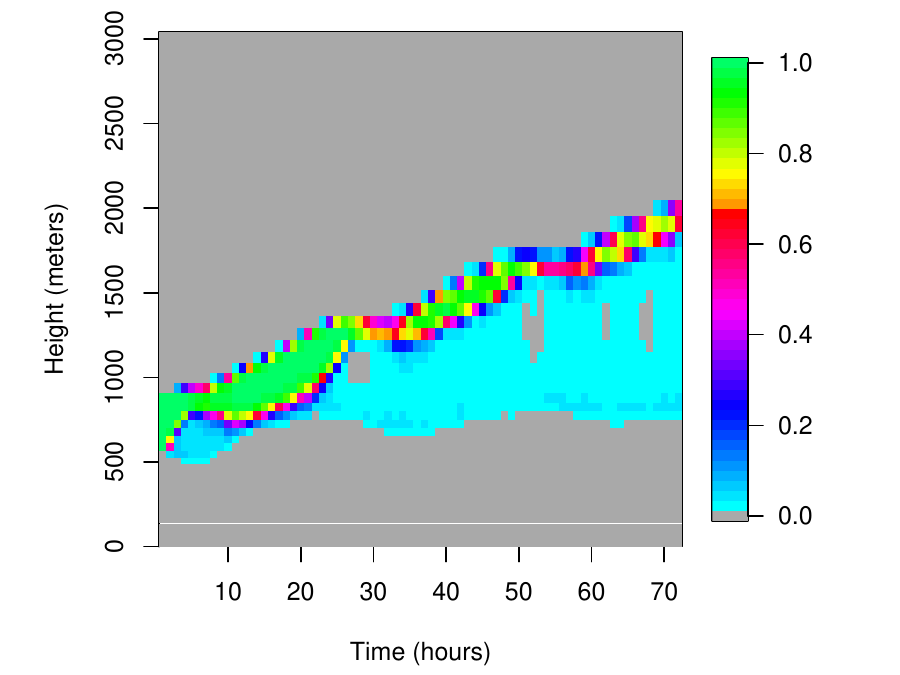}
\caption{Time evolution of the cloud fraction profile over the 3 days of the stratocumulus-to-cumulus transition case study based on the LES used in \cite{Hourdin2019}.}
\label{fig:6.1}
\end{figure}
We ran 90 ensemble members with the single-column model, varying 5 parameters involved in the parametrisation of convection and clouds:  \texttt{thermals\_fact\_epsilon}, \texttt{thermals\_ed\_dz}, \texttt{cld\_lc\_lsc}, \texttt{rad\_chaud1} and \texttt{z0min}. They were varied in their input space (subsequently mapped to $[-1,1]^5$ for emulation and HM) according to a Latin Hypercube design \citep{Processbased1}. The resulting simulations are plotted in S2.1. 

\begin{figure}
\centering
\includegraphics[height=0.35\textheight]{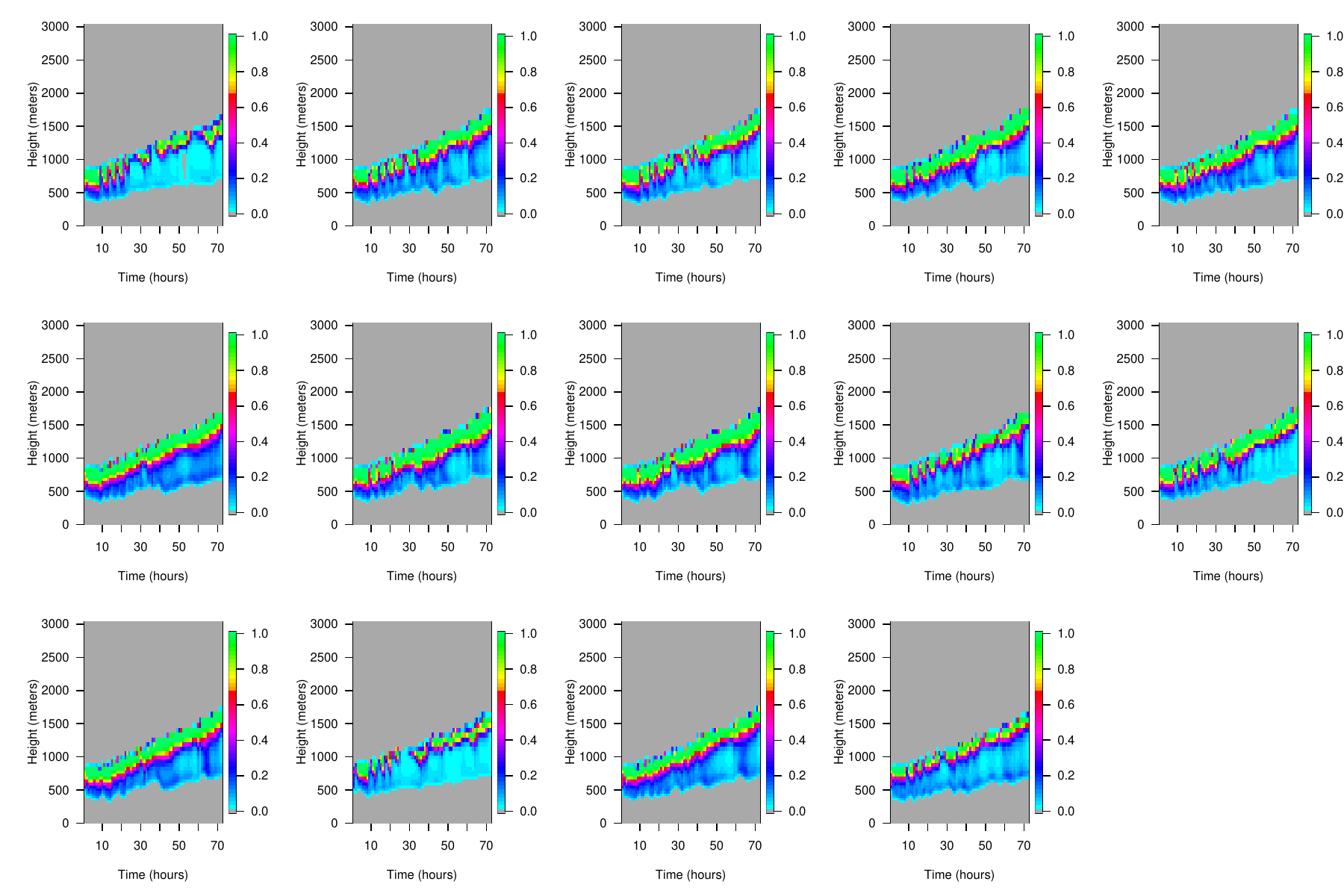}
\caption{The acceptable runs, $\FA$, by expert selection.}
\label{fig:6.3}
\end{figure}
In order to fix our kernels and implausibility thresholds for KHM, we apply the method introduced in Section \ref{sec:selection}, dividing the ensemble into ``acceptable" $(\XA, \FA)$ and ``unacceptable" $(\XU, \FU)$ matches to the reference LES, based on the judgement of co-author F. Hourdin, who coordinates development of the latest atmospheric component of the IPSL climate model. A simple \texttt{R Shiny} application was developed to enable the classification to be made by a series of clicking `accept' or `reject' for the images of the ensemble members. For wave 1, 14 ensemble members were deemed `acceptable' (Figure \ref{fig:6.3}).

\begin{figure}[h!]
		\centering
		\begin{minipage}[t]{0.48\textwidth}
		\includegraphics[height=0.35\textheight,width=1\textwidth]{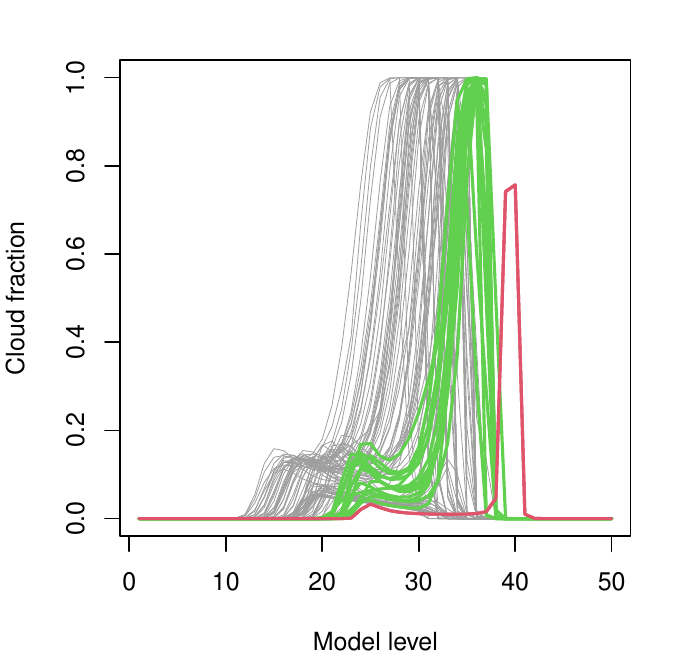}
 	\end{minipage}
 	\begin{minipage}[t]{0.48\textwidth}
		\includegraphics[height=0.35\textheight,width=1\textwidth]{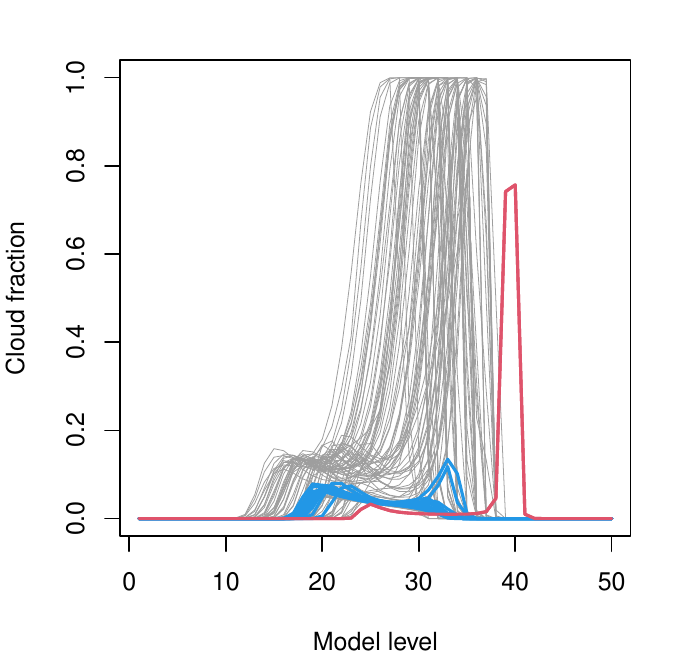}
		\end{minipage}
	\caption{(left) Cloud fraction profile after 68 hours of simulation with the spread of the ensemble of simulations used for wave 1. The wave 1 ensemble is presented in grey, acceptable runs, %selected by the expert, 
	in green and the reference LES in red. (right) Similar to the left panel but with the first 14 wave 1 runs `closest' to the reference LES in blue.}
			\label{fig:6.3.1}
\end{figure}

Though the classification was made solely on the basis of the 2D spatio-temporal images, we can see the types of simulations favoured by looking at individual height profiles. Figure \ref{fig:6.3.1} (left panel) shows the cloud fraction profiles after 68 hours (grey lines) compared to the high resolution simulation (red line), with the acceptable runs highlighted in green. The expert has favoured patterns that lead to higher altitude clouds towards the end of the simulation. Note that the closest simulations (accounting for observation error) to LES in the output space (the blue lines in the right panel of Figure \ref{fig:6.3.1}), are quite different from those runs favoured by the expert, indicating that the standard PCA based calibration methods that compared data in the model output space cannot work, and a kernel strategy is the right one for this problem.

We use the mixture kernel from (\ref{mixture kernel}) with $k_2$ the Gaussian kernel in (\ref{nonlinear kernel}) and weight matrix, $W = \Sigma_e + \Sigma_{\eta}$, where $\Sigma_e$ is the LES uncertainty \cite[see][for its estimate]{Processbased1}, and $\Sigma_{\eta}$ is a positive definite covariance matrix (that would play the role of the model discrepancy if $\omega=1$). Hence $\kpar = \{\omega, \delta, \sigma, \Sigma_{\eta}\}.$ We find these by maximising $\mathcal{P}(\kpar, T^{**})$ in (\ref{c5:performance evaluation}) with $\alpha=0.8$, which gives
$\omega=0.080$, $\delta=0.490$, $\sigma = 428.699$ and the two correlation length parameters in $\Sigma_\eta$ (defined as Gaussian kernel) are suggested as $0.394$ and $0.255$.

We calculate the ensemble projections with KPCA using (\ref{projection}) and retain the first 5 KPCA coefficients following \cite{higdon2008computer}. We fit Gaussian Process emulators to each coefficient, $C_i(\x)$, using the \texttt{tgp} package in \texttt{R}. Leave-one-out cross-validation plots are shown in Figure \ref{appendixchap6.1} where the black error bars represent the emulator 95\% prediction intervals, and the green points are the predicted ensemble members.

\begin{figure}
\centering
\includegraphics[height=0.35\textheight]{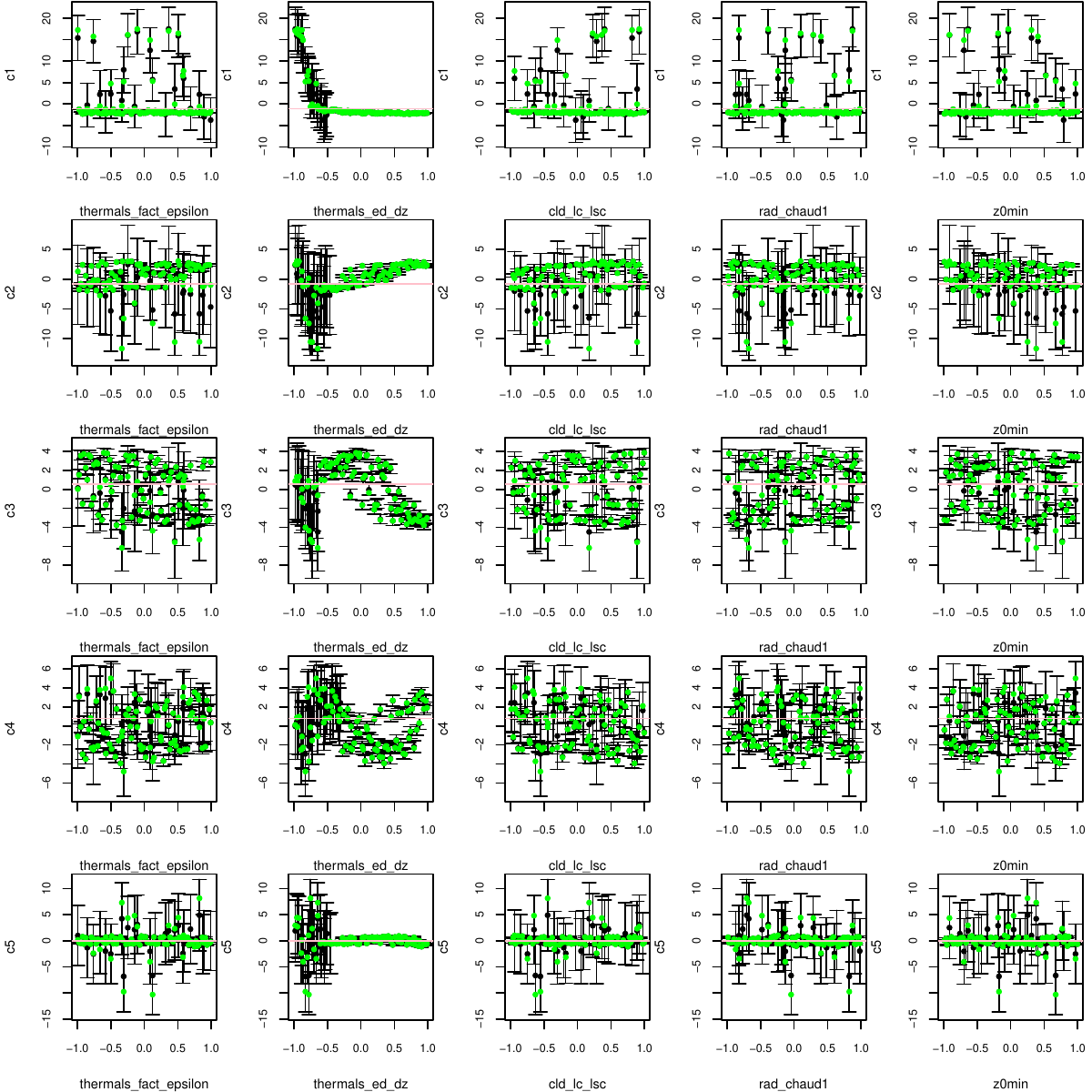}
\caption{Leave-one-out cross-validation plots: wave 1 Gaussian process emulators for $\textbf{C}(\x)$. Black  error bars represent the emulator 95\% prediction intervals, and the green points are the predicted ensemble members}
\label{appendixchap6.1}
\end{figure}

We use KHM with both $\impf{1}{\x}$ and $\impf{2}{\x}$ to rule out implausible regions of parameter space. The wave 1 NROY density plots and the minimum implausibility plots for each pair of parameters is shown in Figure \ref{appendixchap6.2}. We achieve an NROY space $\mathcal{X}^1$ of size $57.78\%$ of $\mathcal{X}$ with $\impf{1}{\x}$, $46.63\%$ of $\mathcal{X}$ with $\impf{2}{\x}$. Figure \ref{appendixchap6.2} shows that NROY space given by $\impf{1}{\x}$ is generally similar to the space given by $\impf{2}{\x}$. KHM with $\impf{1}{\x}$ retains more input space in this example, though this is not always the case \citep[see][for examples]{xu2021generalising}. We select $\impf{1}{\x}$, the more conservative implausibility in this case for subsequent waves.

\begin{figure}
		\centering
		\begin{minipage}[t]{0.48\textwidth}
		\includegraphics[height=0.35\textheight,width=1\textwidth]{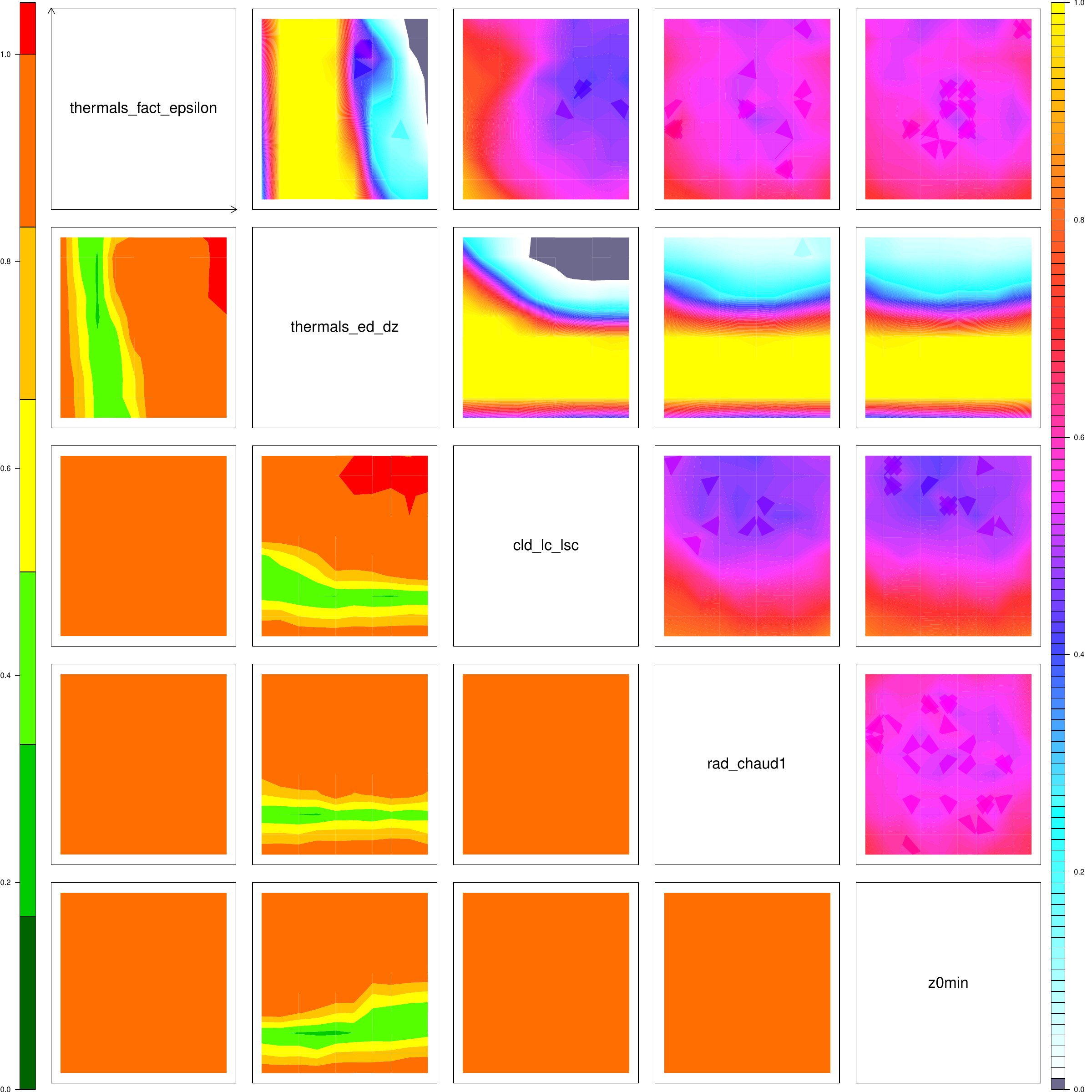}
 	\end{minipage}
 	\begin{minipage}[t]{0.48\textwidth}
		\includegraphics[height=0.35\textheight,width=1\textwidth]{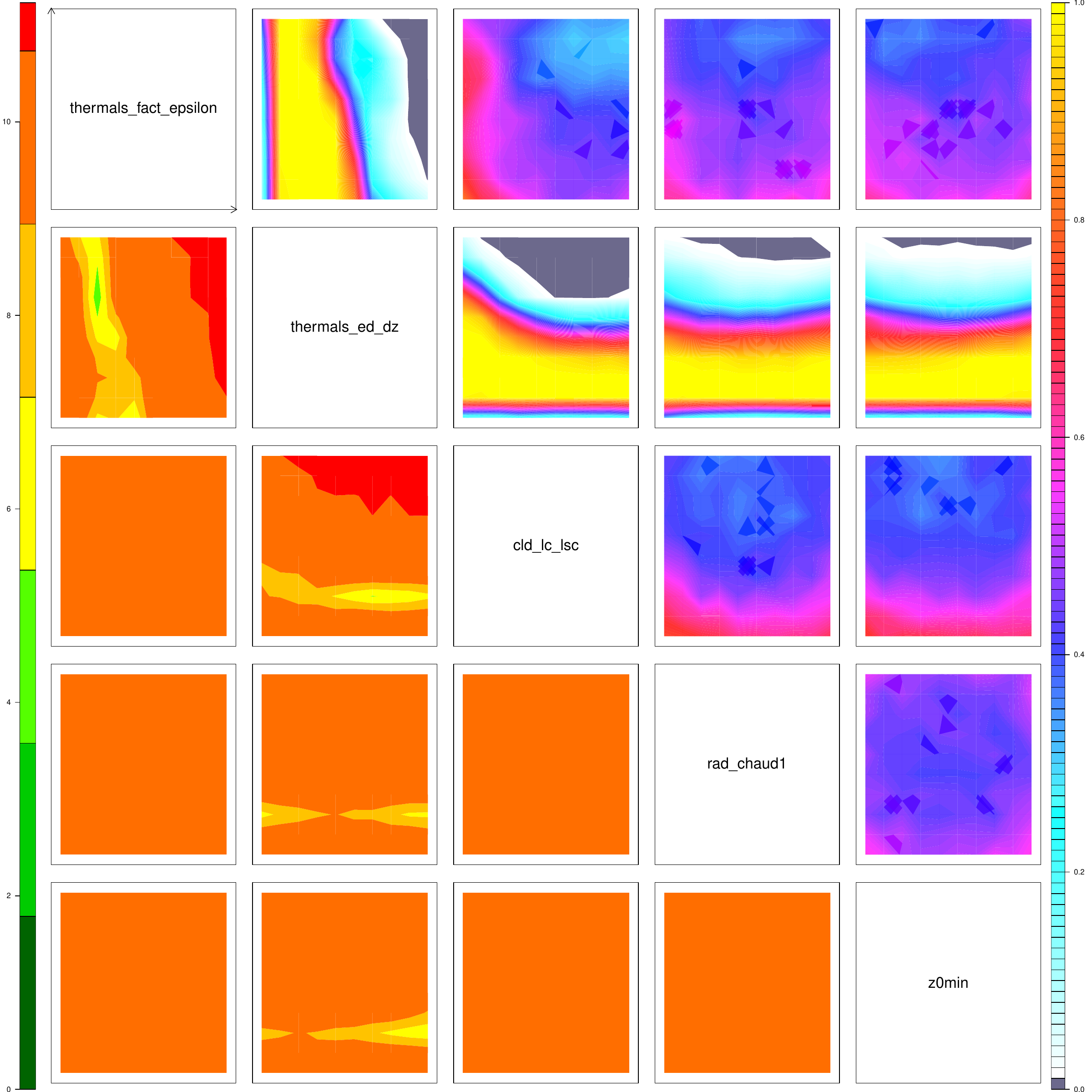}
		\end{minipage}
\caption{Wave 1 with $\impf{1}{\x}$(\textit{left}) and wave 1 with $\impf{2}{\x}$ (\textit{right}) NROY density plots. \textit{Upper triangle:} NROY density plots for each pair of parameters. \textit{Lower triangle:} minimum implausibility plots for each pair of parameters.}
\label{appendixchap6.2}
\end{figure}

\subsection{Refocusing: wave 2 \& wave 3}
We repeat the above KHM exercise for 2 further waves, each time selecting 90 ensemble members from the current NROY space, performing a new expert classification and subsequent kernel estimation, then emulating the first 5 KPCA coefficients and using $\impf{1}{\x}$ to rule out parameter space. As discussed in Section \ref{sec::acceptable}, given each new ensemble, the expert can view the existing simulations and change what they mean by `acceptable' and `unacceptable' representations of key features. In wave 2 only 9 runs were classified as acceptable, and 12 were chosen in wave 3 (these are shown in supplementary S2.2).

The wave 2 and 3 ensemble members are shown in S2.1, and we plot the cloud fraction profiles after 68 hours across all 3 waves, with wave 1 in grey, wave 2 in blue and wave 3 in yellow, in Figure \ref{fig:6.11}. The fitted wave 2 kernel parameters were 
$\omega=0.0970$, $\sigma = 0.285$ and the two correlation length parameters in $\Sigma_\eta$ are suggested as $0.025$ and $0.023$, and wave 2 kernel parameters were $\omega=0.901$, $\sigma = 0.018$ and the two correlation length parameters in $\Sigma_\eta$ are suggested as $0.060$ and $0.168$.
There is a large difference between the wave 1 and wave 2+3 kernels, potentially indicating a shift in expert classification as better simulations become available. Only $28.36\%$ of parameter space is NROY after wave 2 and this is still $26.12\%$ after wave 3, indicating that there is no need for further waves (unless the expert wanted to refine their feature selection to be more discriminatory, or to introduce new metrics). The NROY spaces for wave 2 and 3 are shown in Figure \ref{fig:6.10}.

\begin{figure}
\centering
\includegraphics[height=0.35\textheight]{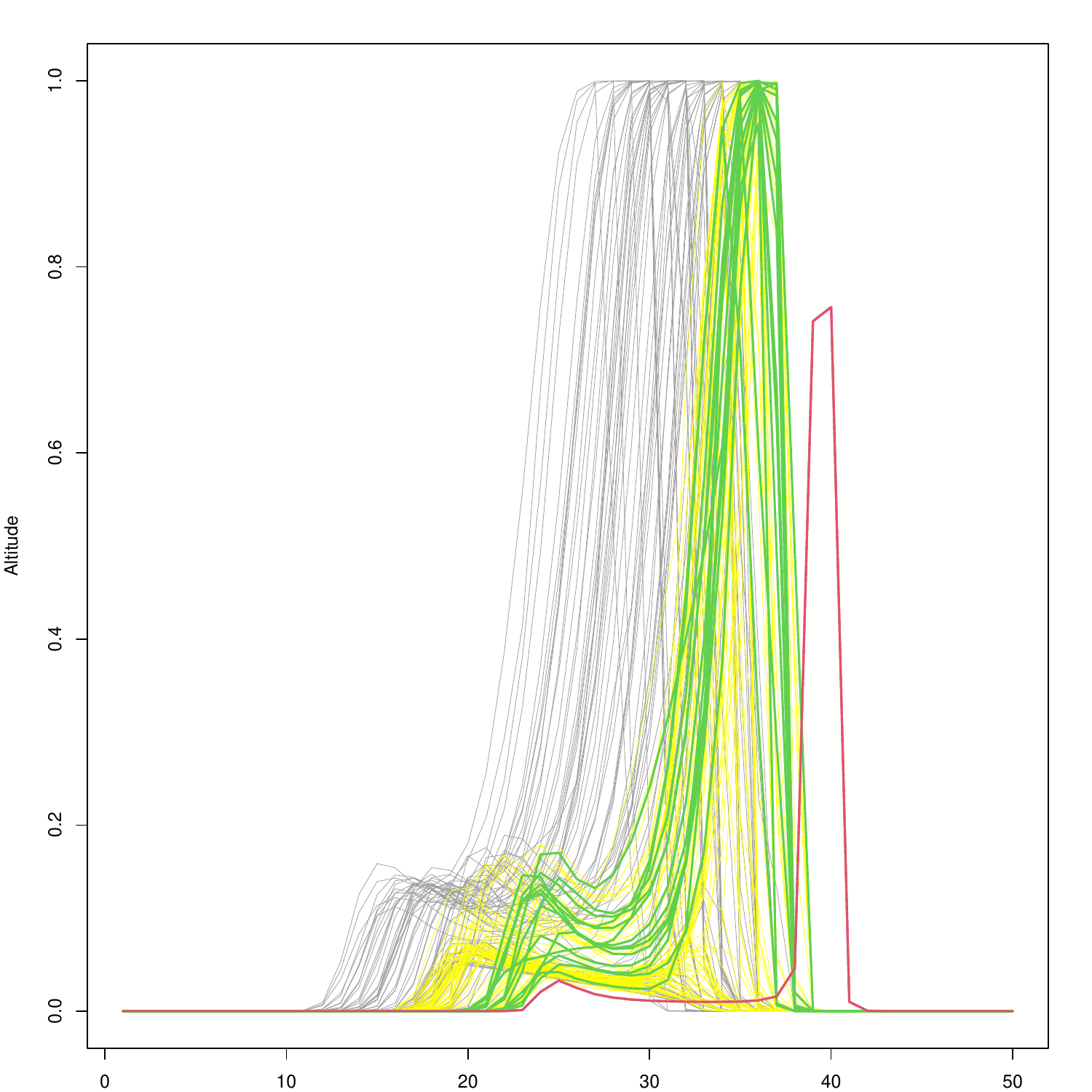}
\caption{The cloud fraction for the later hour (time=68) of the simulation with the spread of the ensemble of simulations used for the different waves indicated in different colours. The wave 1 ensemble is presented in grey, the wave 2 ensemble is presented in yellow, wave 3 ensemble is presented in green, and the reference LES in thick red. }
\label{fig:6.11}
\end{figure}

\begin{figure}
\centering
\includegraphics[height=0.33\textheight, width=0.49\textwidth]{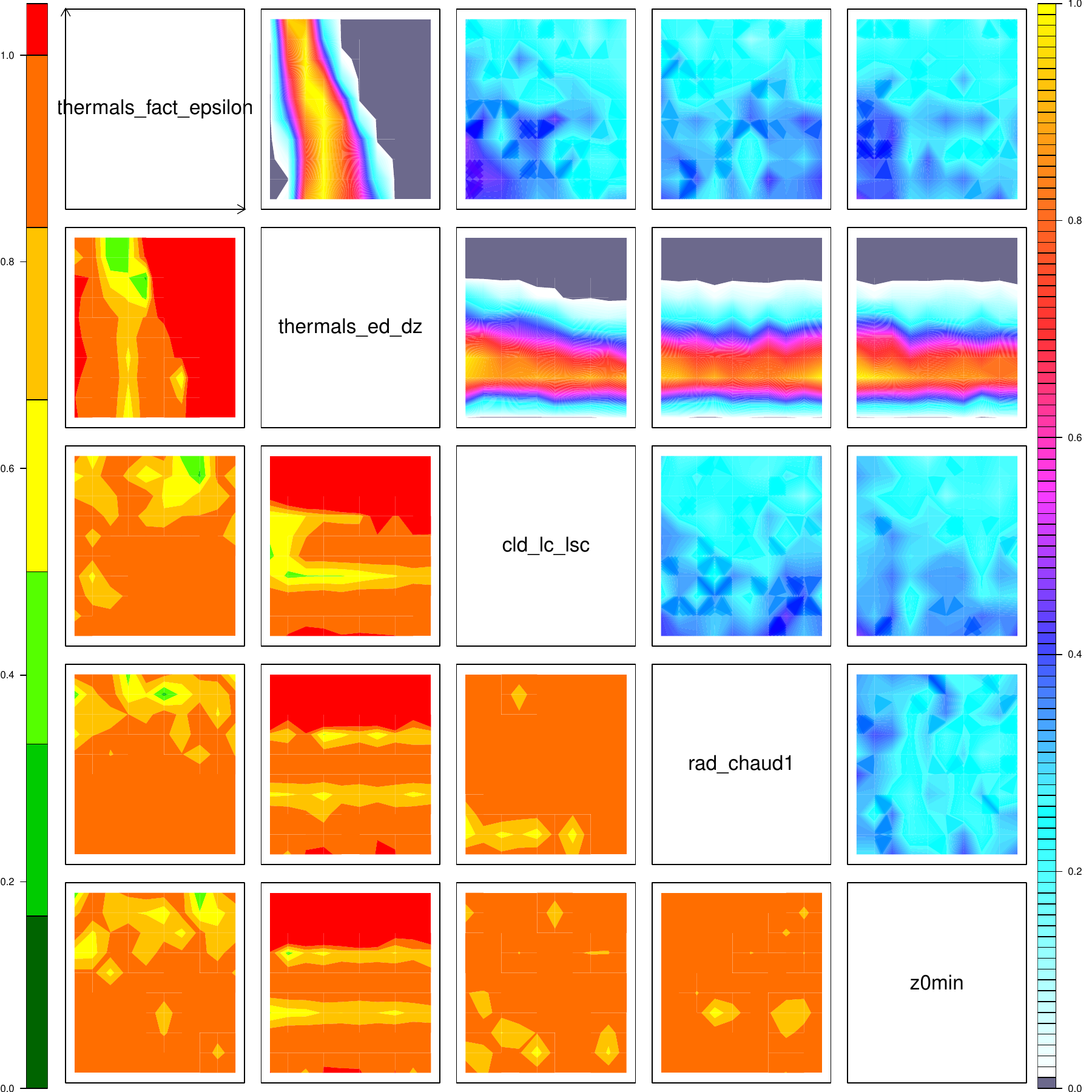}
\includegraphics[height=0.33\textheight, width=0.49\textwidth]{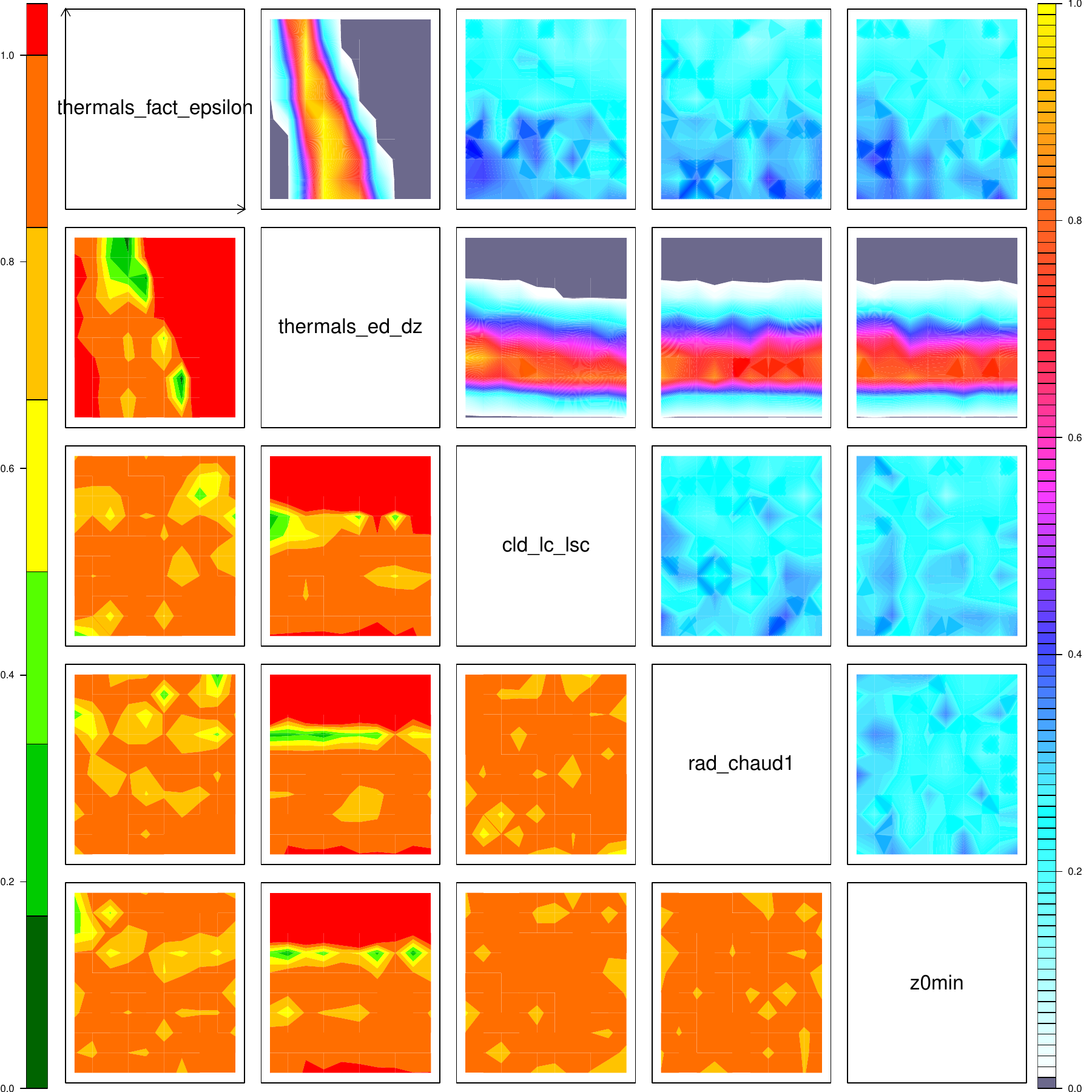}
\caption{Wave 2 (\textit{left}) and wave 3 (\textit{right}) NROY density plots (\textit{upper triangles}) and minimum implausibility plots (\textit{lower triangles}), as described in the caption for Figure \ref{appendixchap6.2}  }
\label{fig:6.10}
\end{figure}

\section{Discussion}
\label{Conclusion}

We have generalised History Matching for high-dimensional outputs with PCA to feature spaces characterised by kernels via KPCA. This generalisation allows emergent spatio-temporal features in models to be appropriately compared with their real-world counter parts, when, for reasons dependent on the model hypotheses, we should not expect the modelled features to be exactly the same as in reality. Our method re-interprets History Matching's geometric idea of comparing model output and data in the Hilbert space with the covariance inner product to one of comparing them in a general Hilbert Space. We provide theory for history matching in these spaces via 2 types of implausibility: one represented by the distance between the emulator mean and data in the reproducing kernel Hilbert space, where emulator uncertainty is accounted for in the implausibility cutoff bound, and the other building emulator uncertainty into a custom implausibility that converges towards the first when the emulator uncertainty reduced to $0$, and that has a constant cut off threshold.

We suggest a particular form of mixture kernel for KHM that recovers standard history matching when features are comparable in a linear subspace, and advocate fitting this kernel and deriving appropriate implausibility bounds using a partial classification of the training data by the modellers. We believe that additional input from the modellers, beyond providing ensembles and data, is critical to any method attempting to calibrate a model based on complex spatio-temporal fields. There are inevitably patterns unique to certain ensemble members, and patterns in the data that would not be considered important by a modeller. Conversely, what characterises a reasonable representation of an important feature (and indeed what an important feature is) is entirely application and model specific. Our view is that a partial classification of existing ensemble members is not a particular challenge for a modeller, and we have demonstrated its efficacy in our example based on boundary-layer cloud modelling.

A particular strength of the KHM approach is the iterative wave structure that allows initial classification and searches for good parameters to allow a wider set of `acceptable' representations than the modeller might finally want. For example, if no ideal matches are found in wave 1, the expert can still choose the best representations thus-far and then refine their classification in later waves if better matches occur. It should be noted that KHM can be used for calibration to multiple metrics, as is common for climate models for example, where spatio-temporal patterns, assessed with the implausibilities described in this paper, are mixed with other types of metric (e.g. scalars) using the traditional implausibility. Models can then be ruled out if they miss multiple implausibility cutoffs as usual \citep{Processbased1}.

When employing standard history matching, the best input assumption used in Bayesian Calibration (see Equation \ref{best}) holds, so that, after a number of waves of history matching, a sample from the posterior can be obtained (either directly using an MCMC within NROY space or by re-weighting NROY samples as in \cite{salter2016comparison}). KHM has no best input assumption and uses a purely geometric argument to determine how `close' models are to data and how close is `close enough'. This might be formalised in future work to allow a probabilistic model within feature space so that a fully Bayesian version of this approach is also available. Such a model may lead to different implausibility bounds for history matching. 

A final issue that deserves further investigation is the lack of ability to reconstruct the model output using the emulator. In most Uncertainty Quantification problems, the emulator enables fast reconstruction of model output (usually ready for comparison to data), and most basis methods allow fast reconstruction of the field of interest from the emulator via simple matrix multiplication. With the KPCA approach, we generally have no access to the non-linear mapping of the model output, $\phi$, so that whilst our emulator is able to capture the variability in feature space through the coefficients, there is no way to visualise the matches we obtain without running the model. This is the `pre-image' reconstruction problem, well known within the literature on kernel methods. Typical solutions from this area involve reconstructing the observations using nearest neighbours within the ensemble, however, these are unlikely to be appropriate in most Uncertainty Quantification problems where ensemble sizes are typically too small to enable these reconstructions to be accurate.

\appendix

\section{Appendix: Proofs}
\label{sec:Proofs}

\subsection{Proof of Theorem 1}
\label{Implausibility1}

\begin{theorem}
Consider a KPCA emulator for $f(\x)$ truncated to the first $q$ basis vectors so that $\E{\phi(f(\x))} = \Psi_q\E{\textbf{C}_q(\x)}+\bar{\phi}$, with $\E{\textbf{C}_q(\x)}$ the emulator fitted to the first $q$ KPCA coefficients. In this case, 
\begin{equation}
\impf{1}{\x}
=\tilde{k}(z,z)+\E{\textbf{C}_q(x)}^T\E{\textbf{C}_q(x)}-2\E{\textbf{C}_q(x)}^T\textbf{A}\tilde{\textbf{K}}_{z},
\end{equation}
where $\textbf{A}$ is the matrix containing the first $r$ eigenvectors of the centred kernel matrix, $\tilde{\textbf{K}}_{z}=[\tilde{k}(z,f(\x_1)),\tilde{k}(z,f(\x_2)) ,\ \dots  ,\ \tilde{k}(z,f(\x_n))]$.
\end{theorem}

\begin{proof}
Implausibility $\impf{1}{\x}$ can be computed via: 
\begin{equation}
 \begin{aligned}
\impf{1}{\x}
&=\left(\phi(z)-\E{\phi(f(\x))}\right)^T\left(\phi(z)-\E{\phi(f(\x))}\right)\\
&=\left(\phi(z)-(\Psi_q\E{\textbf{C}_q(x)}+\bar{\phi}) \right)^T\left(\phi(z)-(\Psi_q\E{\textbf{C}_q(x)}+\bar{\phi}) \right)\\
&=\left(\phi(z)-\bar{\phi}-\Psi_q\E{\textbf{C}_q(x)} \right)^T\left(\phi(z)-\bar{\phi}-\Psi_q\E{\textbf{C}_q(x)} \right)\\
&=\left(\tilde{\phi}(z)-\Psi_q\E{\textbf{C}_q(x)} \right)^T
\left(\tilde{\phi}(z)-\Psi_q\E{\textbf{C}_q(x)} \right)\\
&=\tilde{\phi}(z)^T\tilde{\phi}(z)+ (\Psi_q\E{\textbf{C}_q(x)})^T(\Psi_q\E{\textbf{C}_q(x)})-2\tilde{\phi}(z)^T(\Psi_q\E{\textbf{C}_q(x)})\\
&=\tilde{k}(z,z)+\left(\sum_{k=1}^{q}\psi_k\E{C_k(\x)}\right)^T
\left(\sum_{k=1}^{q}\psi_k \E{C_k(\x)}\right) -2\tilde{\phi}(z)^T\left(\sum_{k=1}^{q}\psi_k \E{C_k(\x)}\right)\\
&=\tilde{k}(z,z)+\sum_{k=1}^{q} \E{C_k(\x)}^T \psi_k ^T \psi_k  \E{C_k(\x)} -2\tilde{\phi}(z)^T\left(\sum_{k=1}^{q}   \sum_{i=1}^{n}\tilde{\alpha}_{ki}\tilde{\phi}(f(x_i)) \E{C_k(\x)}\right)\\
&=\tilde{k}(z,z)+\sum_{k=1}^{q} \E{C_k(\x)}^T \E{C_k(\x)} -2\left(\sum_{k=1}^{q} \sum_{i=1}^{n}\tilde{\alpha}_{ki} \tilde{\phi}(z)^T\tilde{\phi}(f(x_i)) \E{C_k(\x)}\right)\\
&=\tilde{k}(z,z)+\E{\textbf{C}_q(x)}^T\E{\textbf{C}_q(x)}-2\E{\textbf{C}_q(x)}^T\textbf{A}\tilde{\textbf{K}}_{z},
\label{implausibility-feature02}
    \end{aligned}
\end{equation}
where $\textbf{A}$ is the matrix containing the first $r$ eigenvectors of the centred kernel matrix, $\tilde{\textbf{K}}$, $\Psi_q=(\psi_1,...,\psi_q)$, and 
\begin{equation*}
\tilde{\textbf{K}}_{z}=[\tilde{k}(z,f(\x_1)),\tilde{k}(z,f(\x_2)) ,\ \dots  ,\ \tilde{k}(z,f(\x_n))].
\end{equation*}
\end{proof}

\subsection{Proof of Result 1}
\label{threhsold expert judgement 1}
Before proving Result 1, we first prove the following results:

There is a upper bound, $L(\x^*)$, of the implausibility $\impf{1}{\x^*}$ at the best input $\x^*$.

\begin{proof} $T(\x)$ is a function of $\x$ that accounts for the emulator uncertainty. The implausibility $\impf{1}{\x^*}$ at the best input $\x^*$ can be computed as
\begin{equation}
    \begin{aligned}
\impf{1}{\x^*}
&=\parallel \phi(z)-\E{\phi(f(\x^*))}\parallel^2 \\
&=\parallel \phi(z)-\phi_q(f(\x^*))+\phi_q(f(\x^*))-\E{\phi(f(\x^*))}\parallel^2 \\
&= \parallel \phi(z)-\phi_q(f(\x^*))\parallel^2+\parallel \phi_q(f(\x^*))-\E{\phi(f(\x^*))}
 \parallel^2 \\
&\quad  +2(\phi(z)-\phi_q(f(\x^*)))^T (\phi_q(f(\x^*))\E{\phi(f(\x^*))})\\
&=\parallel \phi(z)-\phi_q(f(\x^*))\parallel^2+\parallel \Psi_q\textbf{C}_q(\x^*)-\Psi_q\E{\textbf{C}_q(\x^*)}\parallel^2\\
&\quad  +2(
\textbf{C}_q(x)^T\textbf{A}\tilde{\textbf{K}}_{z}-\textbf{C}_q(x)^T\textbf{C}_q(x)-\E{\textbf{C}_q(x)}^T\textbf{A}\tilde{\textbf{K}}_{z}+\textbf{C}_q(x)^T\E{\textbf{C}_q(x)})\\
&=\parallel \phi(z)-\phi_q(f(\x^*))\parallel^2+(\textbf{C}_q(\x^*)-\E{\textbf{C}_q(\x^*)})^T\Psi_q^T\Psi_q(\textbf{C}_q(\x^*)-\E{\textbf{C}_q(\x^*)})\\
&\quad  +2(
(\textbf{C}_q(x)-\E{\textbf{C}_q(x)})^T\textbf{A}\tilde{\textbf{K}}_{z}-\textbf{C}_q(x)^T(\textbf{C}_q(x)-\E{\textbf{C}_q(x)}))\\
&= \parallel \phi(z)-\phi_q(f(\x^*))\parallel^2+\parallel \textbf{C}_q(\x^*)-\E{\textbf{C}_q(\x^*)}
 \parallel^2 \\
&\quad  +2(
(\textbf{C}_q(x)-\E{\textbf{C}_q(x)})^T(\textbf{A}\tilde{\textbf{K}}_{z}-\textbf{C}_q(x)))\\
&\leq a +\parallel \textbf{C}_q(\x^*)-\E{\textbf{C}_q(\x^*)}\parallel^2 \\
&=L(\x^*),\\
  \end{aligned}
  \label{kernel-threshold-1}
\end{equation}
where $\phi_q(.)$ is the projection of $\phi$ onto the first $q$ basis function,  the emulator uncertainties are accounted for in $\parallel \textbf{C}_q(\x^*)-\E{\textbf{C}_q(\x^*)}\parallel^2$, and $a$ is an upper bound for $\parallel \phi(z)-\phi_r(f(\x^*))\parallel^2$, which can be set with the expert judgement. To retain all of the acceptable inputs, $\x^{A}$, in the NROY space and rule out all of the unacceptable runs input, $\x^{U}$, we have
\begin{equation}
\begin{aligned}
 a &= \textrm{min} (\parallel \phi(z)-\phi_q(f(\x^{U}))\parallel^2)\\
&=\textrm{min} 
(\tilde{k}(z,z)+\textbf{C}_q(\x^{U})^T\textbf{C}_q(\x^{U})-2\textbf{C}_q(\x^{U})^T\textbf{A}\tilde{\textbf{K}}_{z}).
 \end{aligned}
\end{equation}
\end{proof}

%%%%%%%%%%%% 
%\begin{sres}\label{res2}
The expectation and variance of $\parallel \textbf{C}_q(\x^*)-\E{\textbf{C}_q(\x^*)}\parallel^2$ are 
\begin{equation}
    \E{\parallel \textbf{C}_q(x)-\E{\textbf{C}_q(x)}\parallel^2}=\sum_{k=1}^q \Var{C_k(\x)},
    \label{expectation chi-squared 1}
\end{equation}
\begin{equation}
    \Var{\parallel \textbf{C}_q(x)-\E{\textbf{C}_q(x)}\parallel^2}=2\sum_{k=1}^q \Var{C_k(\x)}^2.
    \label{expectation chi-squared 2}
\end{equation}
%\end{sres}

\begin{proof}
$\textbf{C}_q(x)$ is given by Gaussian process emulators,
\begin{equation*}
    \textbf{C}_q(x)-\E{\textbf{C}_q(x)} \sim N(0,\ \Var{\textbf{C}_q(x)}).
\end{equation*}
For univariate emulators built for each coefficient individually, $\Var{\textbf{C}_q(x)})$ is a diagonal $q \times q$ matrix. For each coefficient, 
\begin{equation*}
\begin{aligned}
&\frac{ (C_k(\x)-\E{C_k(\x)})^2}{\Var{C_k(\x)}} \sim \chi_1^2,\\
&\Rightarrow (C_k(\x)-\E{C_k(\x)})^2 \sim \Var{C_k(\x)}\chi_1^2,\\
\end{aligned}
\end{equation*}
where $\chi_1^2$ is the chi-squared distribution with 1 degree of freedom, and $\E{\chi_1^2}=1$ and $\Var{\chi_1^2}=2$. The expectation and variance of $(C_k(\x)-\E{C_k(\x)})^2$ are therefore
$$\E{(C_k(\x)-\E{C_k(\x)})^2}=\Var{C_k(\x)},$$
and
$$\Var{(C_k(\x)-\E{C_k(\x)})^2}=2\Var{C_k(\x)}^2.$$
So that the expectation and variance of $\parallel \textbf{C}_q(x)-\E{\textbf{C}_q(x)}\parallel^2$ are 
\begin{equation}
    \E{\parallel \textbf{C}_q(x)-\E{\textbf{C}_q(x)}\parallel^2}=\sum_{k=1}^q \Var{C_k(\x)},
    \label{expectation chi-squared 1}
\end{equation}
\begin{equation}
    \Var{\parallel \textbf{C}_q(x)-\E{\textbf{C}_q(x)}\parallel^2}=2\sum_{k=1}^q \Var{C_k(\x)}^2.
    \label{expectation chi-squared 2}
\end{equation}
\end{proof}

\begin{res}\label{res}
Suppose we have fitted Gaussian process emulators for each $C_k(\x), k=1, \ldots, q.$ If the model and data are sufficiently close in feature space, then $$P\left(\impf{1}{\x} > T(\x)\right) < 0.1,$$  where 
\begin{equation}
T(\x)=\sum_{k=1}^q \Var{C_k(\x)}+ 3 \sqrt{2\left(\sum_{k=1}^q \Var{C_k(\x)}\right)^2}+a,
\label{threshold for if1}
\end{equation}
where constant $a$ is an upper bound on $||\phi(z) - \phi_q(f(\xstar))||^2$ that can be set by expert judgement or derived from the ensemble. 
\end{res}

\begin{proof} By letting the threshold for each input $T(\x)$ ($T(\x)>0$) be
\begin{equation}
    T(\x)= \E{L(\x^*)} + t\sqrt{\Var{L(\x^*)}},
\label{Chebyshev's1}
\end{equation}
the Chebyshev–Cantelli inequality (one-sided Chebyshev inequality) gives that
\begin{equation}
\begin{aligned}
    &Pr\left(L(\x^*)-\E{L(\x^*)} > t \sqrt{\Var{L(\x^*)}}\right)\leq \frac{\Var{L(\x^*)}}{\Var{L(\x^*)}+t^2\Var{L(\x^*)}},\\
    &\Rightarrow Pr\left(L(\x^*) > \E{L(\x^*)}+t \sqrt{\Var{L(\x^*)}}\right)\leq \frac{1}{1+t^2},\\
     &\Rightarrow Pr\left(L(\x^*) > T(\x) \right)\leq \frac{1}{1+t^2},\\
\end{aligned}
\label{Chebyshev's2}
\end{equation}
where the value of $t$ ($t>0$) represents that for $L(\x^*)$, at least 
$\frac{t^2}{1+t^2}$ of it's distribution's values are greater than $t$ standard deviations from the mean \citep{hazewinkel2001chebyshev}. By setting $t=3$ for both equations (\ref{Chebyshev's1}) and (\ref{Chebyshev's2}), we have   $T(\x)=\E{L(\x^*)} + 3 \sqrt{\Var{L(\x^*)}}$, and the probability that
$L(\x^*)$ smaller than $T(\x)$ is bigger than $90\%$,
\begin{equation}
Pr \left(L(\x^*) < T(\x) \right) \geq 90\%.
\label{Chebyshev's3}
\end{equation}
For any input $\x$, if $\x=\x^*$, then the probability that it's implausibility, $\impf{1}{\x}$, is bigger than the threshold $T(\x)$ does not exceed $10\%$. Hence, this choice of $t$ implies that we view a value of the implausibility $\impf{1}{\x}$ bigger than threshold $T(\x)$ as 
indicating that it is implausible that $\x=\x^*$. Following proof given before, the threshold, $T(\x)$, given in equation (\ref{Chebyshev's1}) (with t=3) can be computed 
\begin{equation}
T(\x)=\sum_{k=1}^q \Var{C_k(\x)}+ 3 \sqrt{2\left(\sum_{k=1}^q \Var{C_k(\x)}\right)^2}+a,
\label{threshold for if2}
\end{equation}
As the variance of the emulator tends to $0$, the distance between simulator outputs and emulator predictions tends to zero, so that
the threshold function $T(\x)$ will be a constant:
\begin{equation*}
T(\x)= a \qquad \text{as } \qquad
\Var{\phi(f(\x))} \longrightarrow 0.
\end{equation*}
\end{proof}

\subsection{Proof of Theorem 2}
\label{Implausibility2}

Before proving Theorem 2, we first show the calculation of $||\varepsilon_z||^2 $:

\begin{equation}
 \begin{aligned}
 ||\varepsilon_z||^2&=(\tilde{\phi}(z)-\tilde{\phi}_q(z))^T(\tilde{\phi}(z)-\tilde{\phi}_q(z))\\
&=\tilde{k}(z,z)+\textbf{C}_q(z)^T\textbf{C}_q(z)-2\textbf{C}_q(z)^T\textbf{A}\tilde{\textbf{K}}_{z}.
    \end{aligned}
\end{equation}

\begin{proof}
The $k$-th projection $C_k(z)$ is given by:
\begin{equation}
    \begin{aligned}
    C_k(z)
    &=\psi_k^T\tilde{\phi}(z)\\
    &=\sum_{i=1}^{n}\tilde{\alpha}_{ki}\tilde{\phi}(f(\x_i))^T[\phi(z)-\bar{\phi}]\\
 &= 
 \sum_{i=1}^{n}\tilde{\alpha}_{ki}[\phi(f(\x_i))-\bar{\phi}]^T[\phi(z)-\bar{\phi}]\\
 &=\sum_{i=1}^{n}\tilde{\alpha}_{ki}[\phi(f(\x_i))^T\phi(z)-\bar{\phi}^T\phi(z)-\phi(f(\x_i))^T\bar{\phi}+\bar{\phi}^T\bar{\phi}]\\
 &=\tilde{\alpha}_k^T\textbf{K}_{z}-\frac{1}{n}\textbf{1}^T\textbf{K}_{z}(\tilde{\alpha}_k^T\textbf{1})-\frac{1}{n}\tilde{\alpha}_k^T(\textbf{K}\textbf{1})+\frac{1}{n^2}\textbf{1}^T\textbf{K}\textbf{1}(\tilde{\alpha}_k^T\textbf{1}).
    \end{aligned}
\label{obs projection}
\end{equation}
Here $\textbf{K}_{z}=[k(z,f(\x_1)),k(z,f(\x_2)) ,\ \dots  ,\ k(z,f(\x_n)))]$, and $\textbf{1}=[1, \ldots, 1]^T $ is an $n \times 1$ vector. Given the centred kernel matrix, $\tilde{\textbf{K}}_{z}=[\tilde{k}(z,f(x_1)),\tilde{k}(z,f(x_2)) ,\ \dots  ,\ \tilde{k}(z,f(x_n))]$, we can calculate each element $\tilde{k}(z,f(\x))$ by:
\begin{equation*}
    \begin{aligned}
      \tilde{k}(z,f(\x))
      &=  \tilde{\phi}(z)^T \tilde{\phi}(f(\x))\\
      &=[\phi(z)-\bar{\phi}]^T[\phi(f(\x))-\bar{\phi}]\\
      &=k(z,f(\x))-\frac{1}{n}\sum_{i=1}^{n}k(z,f(\x_i))-
      \frac{1}{n}\sum_{i=1}^{n}k(f(\x),f(\x_i))+
      \frac{1}{n^2}\sum_{i=1}^{n} \sum_{j=1}^{n}k(f(\x'),f(\x_i))\\
      &=k(z,f(\x))-\frac{1}{n}\textbf{1}^T\textbf{K}_z-
      \frac{1}{n}\textbf{1}^T\textbf{K}_f(\x)+
      \frac{1}{n^2}\textbf{1}^T\textbf{K}\textbf{1}.
    \end{aligned}
\end{equation*}
Therefore, $\tilde{\textbf{K}}_{z}$ can then be expressed as:
\begin{equation*}
    \begin{aligned}
\tilde{\textbf{K}}_{z}&=[\tilde{k}(z,f(x_1)),\tilde{k}(z,f(x_2)) ,\ \dots  ,\ \tilde{k}(z,f(x_n)))]\\
&=\textbf{K}_{z}-\frac{1}{n}\textbf{1}^T\textbf{K}_{z}\textbf{1}-\frac{1}{n}\textbf{K}\textbf{1}+\frac{1}{n^2}\textbf{1}^T\textbf{K}\textbf{1}\textbf{1}\\
&=\textbf{H}(\textbf{K}_{z}-\frac{1}{n}\textbf{K}\textbf{1}),
    \end{aligned}
\end{equation*}
where $\textbf{H}$ is the centring matrix.

For computation, The $k$-th projection $C_k(z)$ can be written as:
\begin{equation}
       C_k(z)= \psi_k^T\tilde{\phi}(z)= \tilde{\alpha}_k^T\tilde{\textbf{K}}_{z}.
\end{equation}
Hence, every term in Equation (14) can be computed.
\end{proof}

\begin{theorem}
Without necessarily knowing $m$ or having access to mapping function $\phi(\cdot)$, we can calculate $\impf{2}{\x}$ via
\begin{equation}
\impf{2}{\x} = \big(\textbf{C}_q(z)-\E{\textbf{C}_q(\x)}\big)^T\big(\Var{\textbf{C}_q(\x)}+ \textbf{1}_q\big)^{-1}\big(\textbf{C}_q(z)-\E{\textbf{C}_q(\x)}\big) + ||\varepsilon_z||^2,
\label{if2 coefficient way}
\end{equation}
where $\varepsilon_z$ is the observation reconstruction error and \begin{equation}
\begin{aligned}
 ||\varepsilon_z||^2&=(\tilde{\phi}(z)-\tilde{\phi}_q(z))^T(\tilde{\phi}(z)-\tilde{\phi}_q(z))\\
&=\tilde{k}(z,z)+\textbf{C}_q(z)^T\textbf{C}_q(z)-2\textbf{C}_q(z)^T\textbf{A}\tilde{\textbf{K}}_{z},
   \end{aligned}
   \label{varepsilon calculate}
\end{equation}
where $\textbf{C}(z)$ is the projection of $\phi(z)$ onto the coefficient space.
\end{theorem}

\begin{proof}

The proof is a generalisation of the work by \cite{salter2019efficient}, which relies on the well-known Woodbury formula \citep{woodbury1950inverting,higham2002accuracy},
$$(A+UCV)^{-1}=A^{-1}-A^{-1}U\left(C^{-1}+VA^{-1}U\right)^{-1}VA^{-1},$$
for matrices A, U, C and V. Expanding
$\impf{2}{\x}$,

\begin{equation*}    
\begin{aligned}
\impf{2}{\x}
&=\big(\phi(z)-\E{\phi(f(\x))}\big)^T\big(\textbf{1}_m+\Var{\phi(f(\x))}\big)^{-1}
\big(\phi(z)-\E{\phi(f(\x))}\big)\\
&=(\tilde{\phi}(z)-\Psi_q\E{\textbf{C}_q(x)})^T(\textbf{1}_m+\Psi_q\Var{\textbf{C}_q(x)}\Psi_q^T)^{-1}
(\tilde{\phi}(z)-\Psi_q\E{\textbf{C}_q(x)})\\
&=(\tilde{\phi}(z)-\Psi_q\E{\textbf{C}_q(x)})^T
\left(\textbf{1}_m^{-1}-\textbf{1}_m^{-1}\Psi_q(\Var{\textbf{C}_q(x)}^{-1}+\Psi_q^T\textbf{1}_m^{-1}\Psi_q)^{-1}\Psi_q^T\textbf{1}_m^{-1}\right)\\
& \qquad \times(\tilde{\phi}(z)-\Psi_q\E{\textbf{C}_q(x)})\\
&=(\tilde{\phi}(z)-\Psi_q\E{\textbf{C}_q(x)})^T
\left(\textbf{1}_m-\Psi_q(\Var{\textbf{C}_q(x)}^{-1}+\Psi_q^T\Psi_q)^{-1}\Psi_q^T\right)
(\tilde{\phi}(z)-\Psi_q\E{\textbf{C}_q(x)})\\
&=(\tilde{\phi}(z)-\Psi_q\E{\textbf{C}_q(x)})^T
\left(\textbf{1}_m-\Psi_q(\Var{\textbf{C}_q(x)}^{-1}+\textbf{1}_q)^{-1}\Psi_q^T\right)
(\tilde{\phi}(z)-\Psi_q\E{\textbf{C}_q(x)})\\
&=(\tilde{\phi}(z)-\Psi_q\E{\textbf{C}_q(x)})^T
(\tilde{\phi}(z)-\Psi_q\E{\textbf{C}_q(x)})\\
& \qquad -(\tilde{\phi}(z)-\Psi_q\E{\textbf{C}_q(x)})^T
\left(\Psi_q(\Var{\textbf{C}_q(x)}^{-1}+\textbf{1}_q)^{-1}\Psi_q^T\right)
(\tilde{\phi}(z)-\Psi_q\E{\textbf{C}_q(x)}).\\
\end{aligned}
\end{equation*}
Applying the Woodbury formula again to the term of $(\Var{\textbf{C}_q(x)}^{-1}+\textbf{1}_m)^{-1}$, we have:
\begin{equation*}    
\begin{aligned}
(\Var{\textbf{C}_q(x)}^{-1}+\textbf{1}_q)^{-1}
&=\textbf{1}_q^{-1} -\textbf{1}_q^{-1}(\Var{\textbf{C}_q(x)}+\textbf{1}_q^{-1})^{-1}\textbf{1}_q^{-1}\\
&=\textbf{1}_q -(\Var{\textbf{C}_q(x)}+\textbf{1}_q)^{-1}.\\
\end{aligned}
\end{equation*}
Therefore, 
\begin{equation*}    
\begin{aligned}
\impf{2}{\x}
&=(\tilde{\phi}(z)-\Psi_q\E{\textbf{C}_q(x)})^T
(\tilde{\phi}(z)-\Psi_q\E{\textbf{C}_q(x)})\\
&\qquad -(\tilde{\phi}(z)-\Psi_q\E{\textbf{C}_q(x)})^T
\left(\Psi_q(\textbf{1}_q -(\Var{\textbf{C}_q(x)}+\textbf{1}_q)^{-1})\Psi_q^T\right)
(\tilde{\phi}(z)-\Psi_q\E{\textbf{C}_q(x)}).\\
&=(\tilde{\phi}(z)-\Psi_q\E{\textbf{C}_q(x)})^T
(\tilde{\phi}(z)-\Psi_q\E{\textbf{C}_q(x)})\\
&\qquad -(\tilde{\phi}(z)-\Psi_q\E{\textbf{C}_q(x)})^T\Psi_q\Psi_q^T
(\tilde{\phi}(z)-\Psi_q\E{\textbf{C}_q(x)})\\
&\qquad +(\tilde{\phi}(z)-\Psi_q\E{\textbf{C}_q(x)})^T
\Psi_q(\Var{\textbf{C}_q(x)}+\textbf{1}_q)^{-1}\Psi_q^T
(\tilde{\phi}(z)-\Psi_q\E{\textbf{C}_q(x)}).\\
&=(\tilde{\phi}(z)-\Psi_q\E{\textbf{C}_q(x)})^T
(\tilde{\phi}(z)-\Psi_q\E{\textbf{C}_q(x)})\\
&\qquad -(\Psi_q^T\tilde{\phi}(z)-\Psi_q^T\Psi_q\E{\textbf{C}_q(x)})^T
(\Psi_q^T\tilde{\phi}(z)-\Psi_q^T\Psi_q\E{\textbf{C}_q(x)})\\
&\qquad +(\Psi_q^T\tilde{\phi}(z)-\Psi_q^T\Psi_q\E{\textbf{C}_q(x)})^T
(\Var{\textbf{C}_q(x)}+\textbf{1}_q)^{-1}
(\Psi_q^T\tilde{\phi}(z)-\Psi_q^T\Psi_q\E{\textbf{C}_q(x)}).\\
&=(\tilde{\phi}(z)-\Psi_q\E{\textbf{C}_q(x)})^T
(\tilde{\phi}(z)-\Psi_q\E{\textbf{C}_q(x)})\\
&\qquad -(\textbf{C}_q(z)-\E{\textbf{C}_q(x)})^T(\textbf{C}_q(z)-\E{\textbf{C}_q(x)}),\\
&\qquad +(\textbf{C}_q(z)-\E{\textbf{C}_q(x)})^T(\Var{\textbf{C}_q(x)}+ \textbf{1}_q)^{-1}(\textbf{C}_q(z)-\E{\textbf{C}_q(x)}).\\
\end{aligned}
\end{equation*}
We first compute $||\textbf{C}_q(z)-\E{\textbf{C}_q(x)}||^2$, 
\begin{equation}
    \begin{aligned}
||\textbf{C}_q(z)-\E{\textbf{C}_q(x)}||^2
&=(\textbf{C}_q(z)-\E{\textbf{C}_q(x)})^T(\textbf{C}_q(z)-\E{\textbf{C}_q(x)})\\
&=(\Psi_q^T\tilde{\phi}(z)-\E{\textbf{C}_q(x)})^T(\Psi_q^T\tilde{\phi}(z)-\E{\textbf{C}_q(x)})\\
&=(\Psi_q^T\tilde{\phi}(z)-\E{\textbf{C}_q(x)})^T\Psi_q^T\Psi_q(\Psi_q^T\tilde{\phi}(z)-\E{\textbf{C}_q(x)})\\
&=(\Psi_q\Psi_q^T\tilde{\phi}(z)-\Psi_q\E{\textbf{C}_q(x)})^T(\Psi_q\Psi_q^T\tilde{\phi}(z)-\Psi_q\E{\textbf{C}_q(x)})\\
&=(\tilde{\phi}_q(z)-\E{\tilde{\phi}_q(f(\x))})^T(\tilde{\phi}_q(z)-\E{\tilde{\phi}_q(f(\x))})\\
&=(\phi_q(z)-\E{\phi_q(f(\x))})^T(\phi_q(z)-\E{\phi_q(f(\x))}).
    \end{aligned}
\end{equation}
The difference between $||\tilde{\phi}(z)-\Psi_q\E{\textbf{C}_q(x)}||^2$  and $||\textbf{C}_q(z)-\E{\textbf{C}_q(x)}||^2$ is:
\begin{equation}
\begin{aligned}
&||\tilde{\phi}(z)-\Psi_q\E{\textbf{C}_q(x)}||^2-||\textbf{C}_q(z)-\E{\textbf{C}_q(x)}||^2\\
&=(\phi(z)-\E{\phi_q(f(\x))})^T
(\phi(z)-\E{\phi_q(f(\x))})-(\phi_q(z)-\E{\phi_q(f(\x))})^T(\phi_q(z)-\E{\phi_q(f(\x))})\\
&=\tilde{k}(z,z)+\E{\textbf{C}_q(x)}^T\E{\textbf{C}_q(x)}-2
\E{\textbf{C}_q(x)}^T\textbf{A}\tilde{\textbf{K}}_{z}\\
&\qquad -\left((\Psi_q^T\tilde{\phi}(z))^T(\Psi_q^T\tilde{\phi}(z)) + \E{\textbf{C}_q(x)}^T\E{\textbf{C}_q(x)}-2\E{\textbf{C}_q(x)}^T(\Psi_q^T\tilde{\phi}(z))\right)\\
&=\tilde{k}(z,z)-2
\E{\textbf{C}_q(x)}^T\textbf{A}\tilde{\textbf{K}}_{z}-\left((\Psi_q^T\tilde{\phi}(z))^T(\Psi_q^T\tilde{\phi}(z)) -2\E{\textbf{C}_q(x)}^T\textbf{A}\tilde{\textbf{K}}\right)\\
&=\tilde{k}(z,z)-(\Psi_q^T\tilde{\phi}(z))^T(\Psi_q^T\tilde{\phi}(z))\\
&=||\varepsilon_z||^2.
\end{aligned}
\end{equation}

Hence, we have proved that
$$\impf{2}{\x}
= \big(\textbf{C}_q(z)-\E{\textbf{C}_q(\x)}\big)^T\big(\Var{\textbf{C}_q(\x)}+ \textbf{1}_q\big)^{-1}\big(\textbf{C}_q(z)-\E{\textbf{C}_q(\x)}\big) + ||\varepsilon_z||^2,$$
we can now efficiently calculate the implausibility in feature space, $\impf{2}{\x}$, without requiring the explicit form of the mapping function.
\end{proof}

\subsection{Proof of Theorem 3}
\label{Achieving standard history matching with kernel approach}

\begin{theorem}
\label{standardHM}
Consider Kernel history matching using $\impf{2}{\x}$, the mixture kernel given in (19) (main file), and $W$ in (12) (main file) as discussed in Section-3.2. Setting $\omega=1$ recovers standard History matching with the PCA basis and implausibility (2) (main file).
\end{theorem}

\begin{proof}

A feature space which normally lies in a higher dimensional space than that of the original output space is determined by a kernel function. For the linear kernel, $k(f(\x'),f(\x))=f(\x)^T f(\x)$, the feature space is equivalent to the original space, and kernel PCA is exactly equivalent to standard PCA. The explicit formulation of the map function $\phi(\cdot)$ is $\phi(f(x))=f(x)$.

The relationship between kernel PCA and standard PCA suggests that standard history matching can be achieved by KHM. We adopt the form of our mixture kernel defined in equation (19) (main file), and set the weight parameter, $\omega=1$, the kernel function is then 
$$k(f(\x),f(\x’))= f(\x)^T W^{-1}f(\x’),$$
Note that, $W$ is symmetric and can be decomposed as $W= Q^T H Q$, where $Q$ is an orthogonal matrix, and $H$ is a diagonal matrix with the eigenvalues of $W$ as its diagonal elements. Hence, $W^{-1}$ can be written as $W^{-1}= Q^T H^{-1} Q$. Further, by writing $P=Q^TH^{-1/2}$, we have $W^{-1}= PP^T$. The linear kernel is then
$$k(f(\x),f(\x’))= \phi(f(\x))^T\phi(f(\x)^{'})=f(\x)^T PP^T f(\x’).$$
Since we are using the linear kernel for kernel PCA, the feature space is equivalent to the model output space, and the dimension of mapped data, $m$, is the same as the dimension of model output $l$. The explicit formulation of the map function in this case is:
$$\phi(f(\x))=P^T f(\x).$$ 
Therefore, performing kernel PCA on the original dataset $\textbf{F}=(f(\x_1) ,\ \dots  ,\ f(\x_n))$ is equivalent to performing PCA with the $l \times n$ matrix of mapped ensemble members
$$\Phi=(\phi(f(\x_1)) ,\ \dots  ,\ \phi(f(\x_n))=(P^T f(\x_1) ,\ \dots  ,\  P^T f(\x_n)).$$
The mapped ensemble mean, $\bar{\phi}$, given by averaging across the rows of $\Phi$, could be expressed as the original ensemble mean $u$, with 
$$\bar{\phi}=P^T u.$$
Given the centred mapped ensemble $\tilde{\Phi}$, the PCA/SVD basis, $\Psi$, can be calculated via
$$\tilde{\Phi}^T= E B\Psi^T,$$
where $E$ is a $n\times n$ orthonormal
matrix, $\Psi$ is a $m \times m$ orthonormal matrix, and $B$ is a $n \times m$ matrix with non-zero elements only along the main diagonal. 
$\Psi$ is truncated after the first $r$ vectors, giving the truncated basis, $\Psi_q=(\psi_1  ,\ \dots  ,\ \psi_{q})$. The projection of a mapped output $\phi(f(\x))=P^T f(\x)$ onto $\Psi_q$ is given by
\begin{equation}
    \textbf{C}_q(\x)= (\Psi_q^T \Psi_q)^{-1}\Psi_q^T(\phi(f(\x))-\bar{\phi}),
\end{equation}
where $\textbf{C}_q(\x)=[C_1(\x) ,\ \dots  ,\ C_q(\x) ]^T$, and this projection is exactly the same as the projections for PCA-based history matching given in \cite{salter2019efficient}:
\begin{equation}
\begin{aligned}
  \textbf{C}_q(\x)
  &= (\Psi_{PCA_q}^T PP^T \Psi_{PCA_q})^{-1}\Psi_{PCA_q}^T P(\phi(f(\x))-\bar{\phi})\\
  &= (\Psi_{PCA_q}^T PP^T \Psi_{PCA_q})^{-1}\Psi_{PCA_q}^T PP^T(f(\x)-u)\\
  &=(\Psi_{PCA_q}^T W^{-1}\Psi_{PCA_q})^{-1}\Psi_{PCA_q}^TW^{-1}(f(\x)-u),\\
\end{aligned}
\label{projections pcakpca}
\end{equation}
where $\Psi_{PCA}$ is the PCA basis,  $\Psi_{PCA_q}$ is the truncated basis. $\Psi_{PCA}$ is defined using the original model outputs from the ensemble,
$$F^T= E' B' \Psi_{PCA}^T.$$
Here, $E'$ is a $n\times n$ orthonormal
matrix, $\Psi_{PCA}$ is an $l \times l$ orthonormal matrix, and $B'$ is an $n \times l$ matrix with non-zero elements only along the main diagonal. By the properties of the singular value decomposition, and as $\phi(f(\x))=P^T f(\x)$, we have that $\Psi= P^T \Psi_{PCA}^T$.

%A set of coefficients $\textbf{C}_q(\x)$ can be used to find the reconstruction of $\phi(f(\x))$, with
%\begin{equation}
%    \phi_q(f(\x)) =\Psi_q\Psi_q^T P^T (f(\x)-u)+ P^T u
%\end{equation}

To test whether KHM is the same as standard history matching based on the proposed kernel function, we compare the implausibility functions. When there is no emulator, the implausibility for standard history matching, $\mathcal{I}(\x)$, is defined as the Mahalanobis distance between the observations and the computer model:
$$\mathcal{I}(\x)=(z-f(\x))^T (\Sigma_e+\Sigma_\eta)^{-1} (z-f(\x)).$$
For the same input $\x$, the implausibility for KHM,  $ \mathcal{I}_{0}(\x)$, is defined as the euclidean distance between the observations and the computer model:
$$ \mathcal{I}_{0}(\x)
=\left(\phi(z)-\phi(f(\x))\right)^T\left(\phi(z)-\phi(f(\x))\right),$$
which can be written as:
\begin{equation}
    \begin{aligned}
   \mathcal{I}_{0}(\x)
&=(P^T z-P^T f(\x) )^T(P^T z-P^T f(\x) )\\
&=(z-f(\x))^T P P^T (z-f(\x))\\
&=(z-f(\x))^T W^{-1} (z-f(\x))\\
&=(z-f(\x))^T (\Sigma_e+\Sigma_\eta)^{-1} (z-f(\x)).\\
    \end{aligned}
    \label{st=kerhm}
\end{equation}
Hence, the implausibility for kernel PCA-based history matching and PCA history matching is the same.

When the emulators are required and built for the coefficients on the first $q$ basis vectors, $\textbf{C}_q(\x)=[C_1(\x) ,\ \dots  ,\ C_q(\x) ]^T$,
\begin{equation}
     C_k(\x) \sim \mathcal{GP}(m_k(\x), \ \sigma_k^2 c_k(\x,\x)),  \quad  k=1 ,\ \dots  ,\ r,
\end{equation}
with the emulator expectation for each of the $r$ basis vectors given by 
$$ \E{\textbf{C}_q(\x)}=[\E{C_1(\x)} ,\ \dots  ,\ \E{C_q(\x)} ]^T,$$
and the associated emulator variance matrix: 
$$\Var{\textbf{C}_q(\x)}= \text{diag}[\Var{C_1(\x)} ,\ \dots  ,\ \Var{C_q(\x)}].$$
Note that the coefficients are same for PCA-based history matching and kernel PCA-based history matching. Thus, we 
retrieve the expectation and variance of $f(\x)$ via:
\begin{equation}
\E{f(\x)} = \Psi_{PCA_q} \E{\textbf{C}_q(\x)}+u, \quad 
\Var{f(\x)} = \Psi_{PCA_q}\Var{\textbf{C}_q(\x)}\Psi_{PCA_q}^T,
\end{equation}
and the expectation and variance of $\phi(f(\x))$ via:
\begin{equation}
    \E{\phi(f(\x))} = \Psi_q\E{\textbf{C}_q(\x)}+\bar{\phi},
\quad 
    \Var{\phi(f(\x))} = \Psi_q\Var{\textbf{C}_q(\x)}\Psi_q^T.
\end{equation}
Since $\Psi= P^T \Psi_{PCA}^T$ and $\bar{\phi}=P^T u$, we then can write the relationship between $\E{f(\x)}$ and $\E{\phi(f(\x))}$:
$\E{\phi(f(\x))}=P^T \E{f(\x)}$,
and the the relationship between $\Var{f(\x)} $ and $\Var{\phi(f(\x))}$:
$ \Var{\phi(f(\x))}= P^T \Var{f(\x)} P$.
Based on these relationships, we prove that the implausibility for standard history matching is the same the implausibility for KHM, we have:
\begin{equation}
\mathcal{I}(\x)= \mathcal{I}_{0}(\x),
\label{theorem1}
\end{equation}
where
$$\mathcal{I}(\x)=\big(z-\E{f(\x)}\big)^T \big(\Sigma_e+\Sigma_\eta+\Var{f(\x)}\big)^{-1} \big(z-\E{f(\x)}\big),$$
and
$$\impf{2}{\x}
=\big(\phi(z)-\E{\phi(f(\x))}\big)^T \big(\textbf{1}_m+ \Var{\phi(f(\x))}\big)^{-1}
\big(\phi(z)-\E{\phi(f(\x)}\big).$$

\end{proof}

\section{Appendix: Figures for Process-tuning for climate models}
\subsection{Ensemble plots}

\begin{figure}[H]
        \centering
        \begin{minipage}[t]{1\textwidth}
      \includegraphics[width=1.1\textwidth]{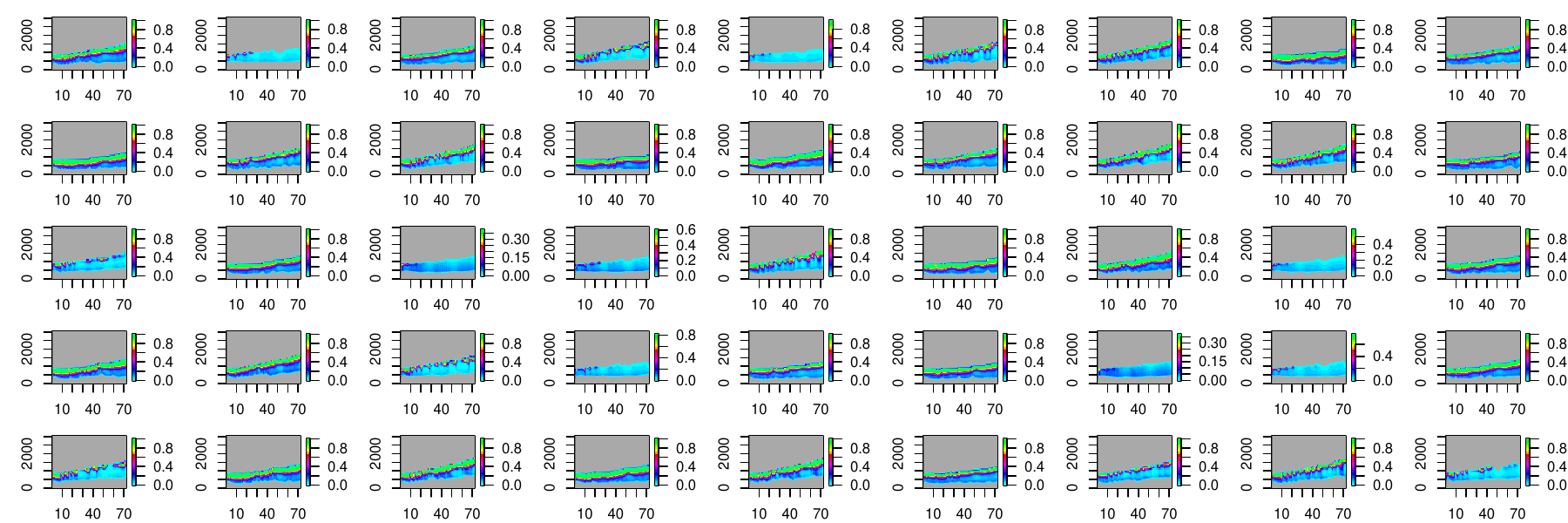}
     \end{minipage}
     \begin{minipage}[t]{1\textwidth}
        \includegraphics[width=1.1\textwidth]{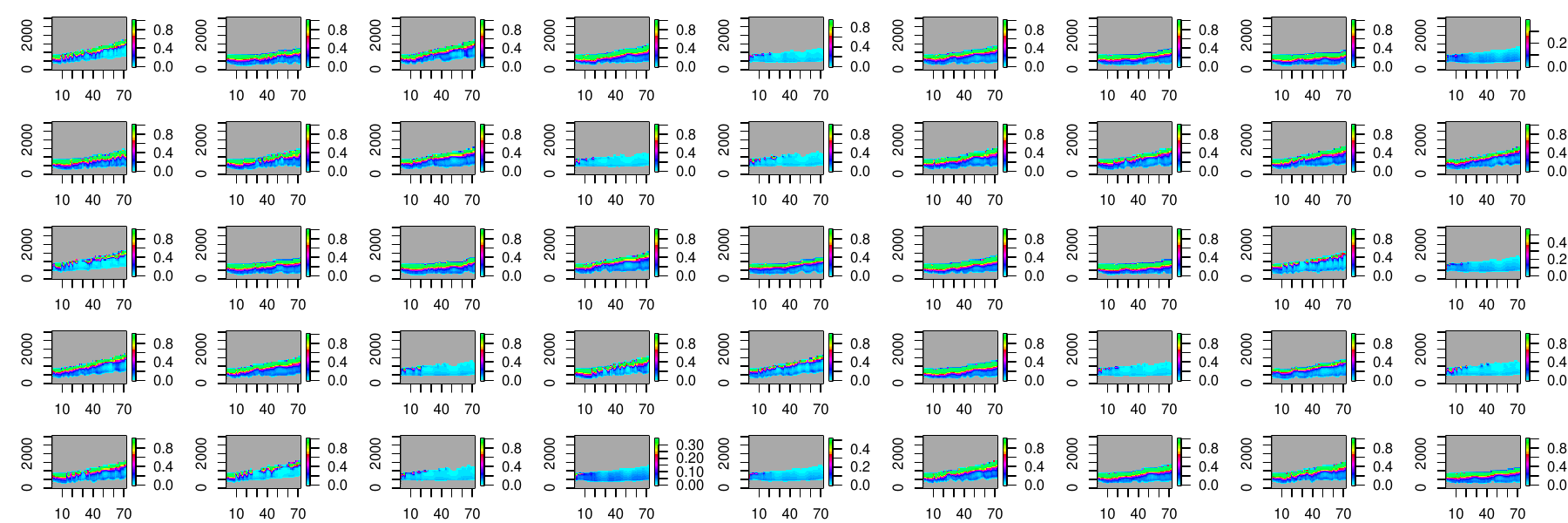}
        \end{minipage}
\caption{Wave 1 Ensemble runs from SCM simulators: the ensemble outputs are plotted ordinarily from the 1st run to the 90th. For each plot, it shows the hourly averages of the cloud fraction profiles during 72 hours of SCM simulation. }
\label{appendix:wave1}
\end{figure}

\begin{figure}[H]
        \centering
        \begin{minipage}[t]{1\textwidth}
      \includegraphics[width=1.1\textwidth]{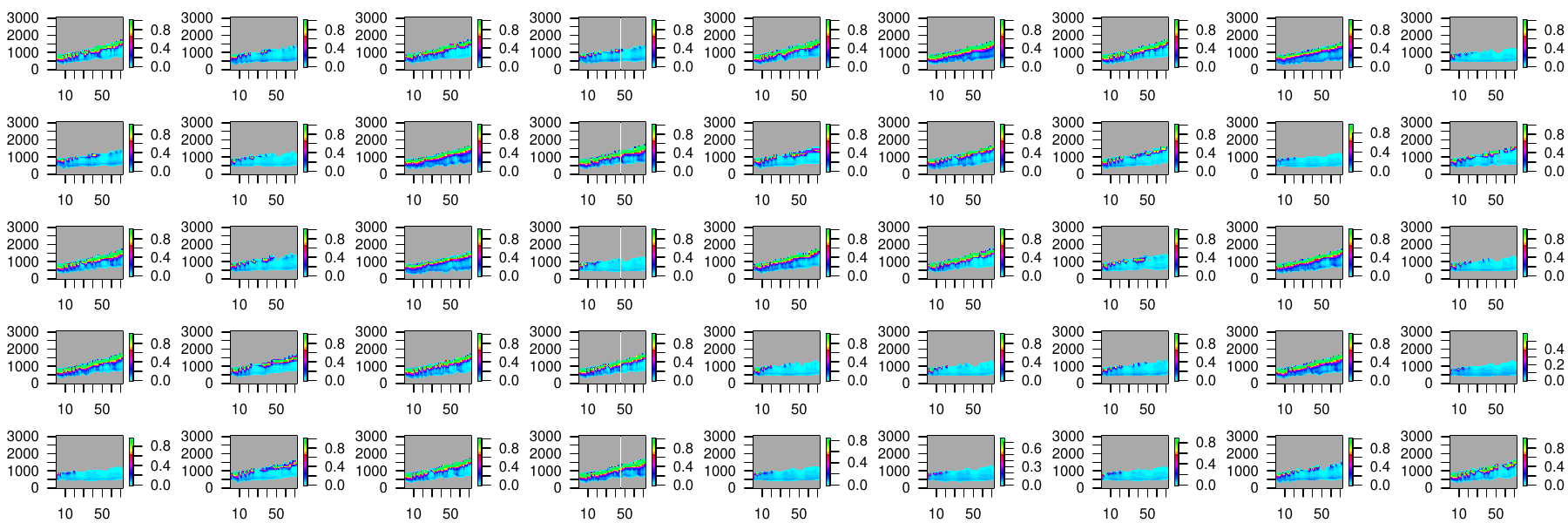}
     \end{minipage}
     \begin{minipage}[t]{1\textwidth}
        \includegraphics[width=1.1\textwidth]{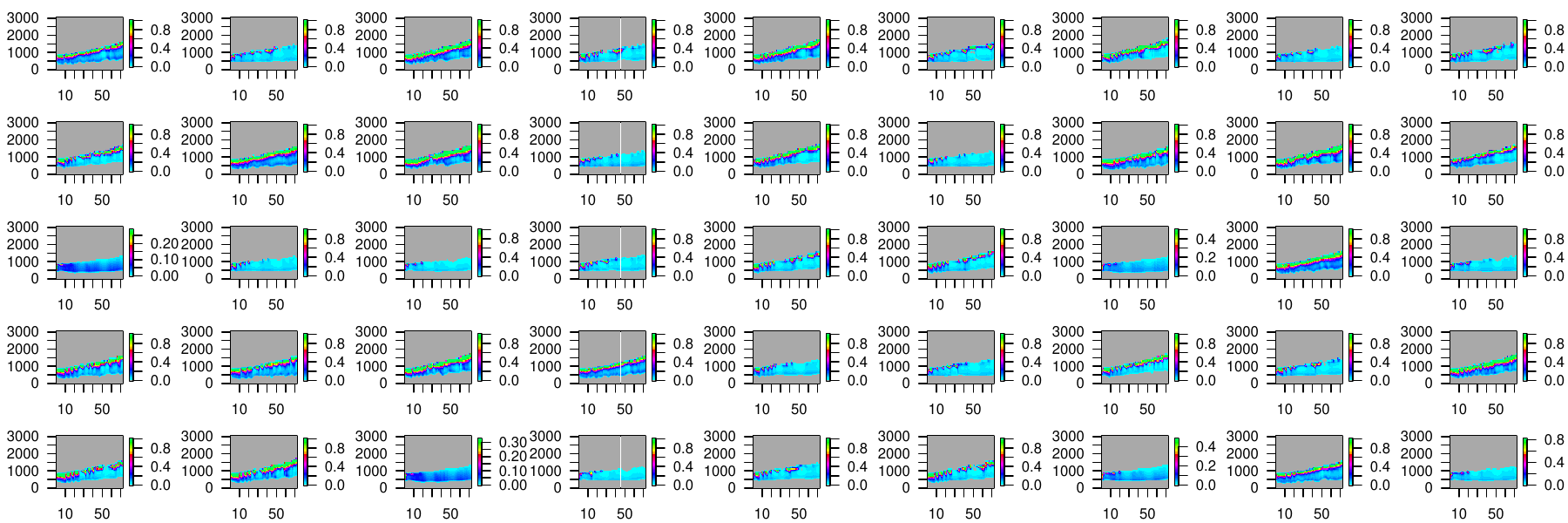}
        \end{minipage}
\caption{Wave 2 Ensemble runs from SCM simulators: the ensemble outputs are plotted ordinarily from the 1st run to the 90th.}
\label{appendix:wave2}
\end{figure}

\begin{figure}[H]
        \centering
        \begin{minipage}[t]{1\textwidth}
      \includegraphics[width=1.1\textwidth]{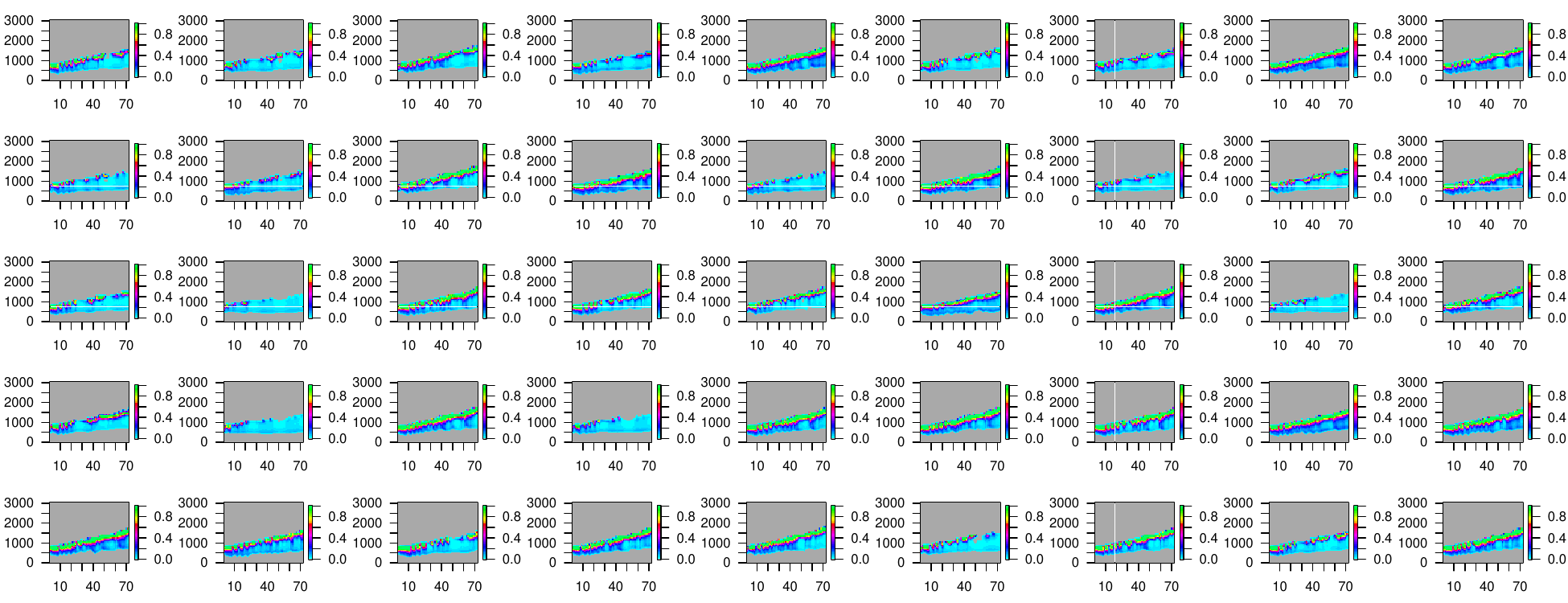}
     \end{minipage}
     \begin{minipage}[t]{1\textwidth}
        \includegraphics[width=1.1\textwidth]{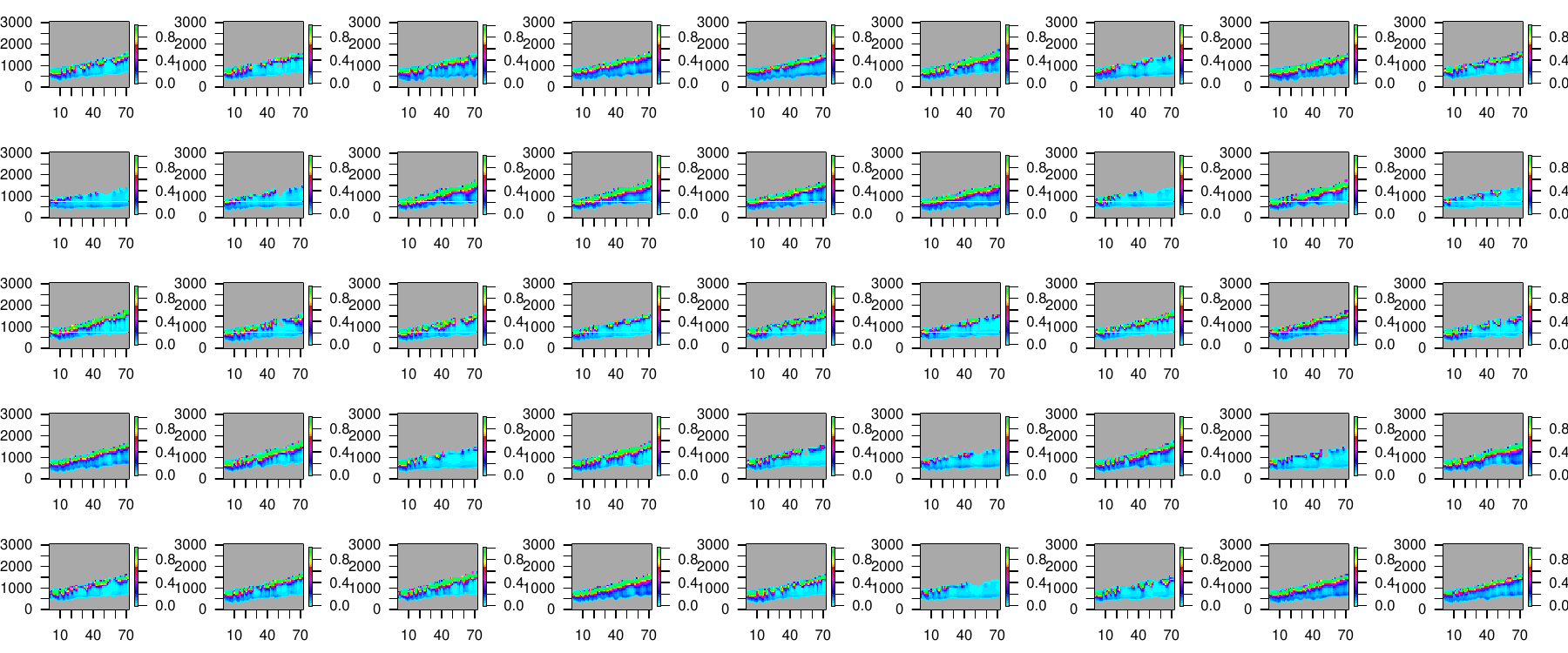}
        \end{minipage}
\caption{Wave 3 Ensemble runs from SCM simulators: the ensemble outputs are plotted ordinarily from the 1st run to the 90th.}
\label{appendix:wave3}
\end{figure}

\subsection{Acceptable runs plots}

\begin{figure}[H]
\centering
\includegraphics[width=14cm]{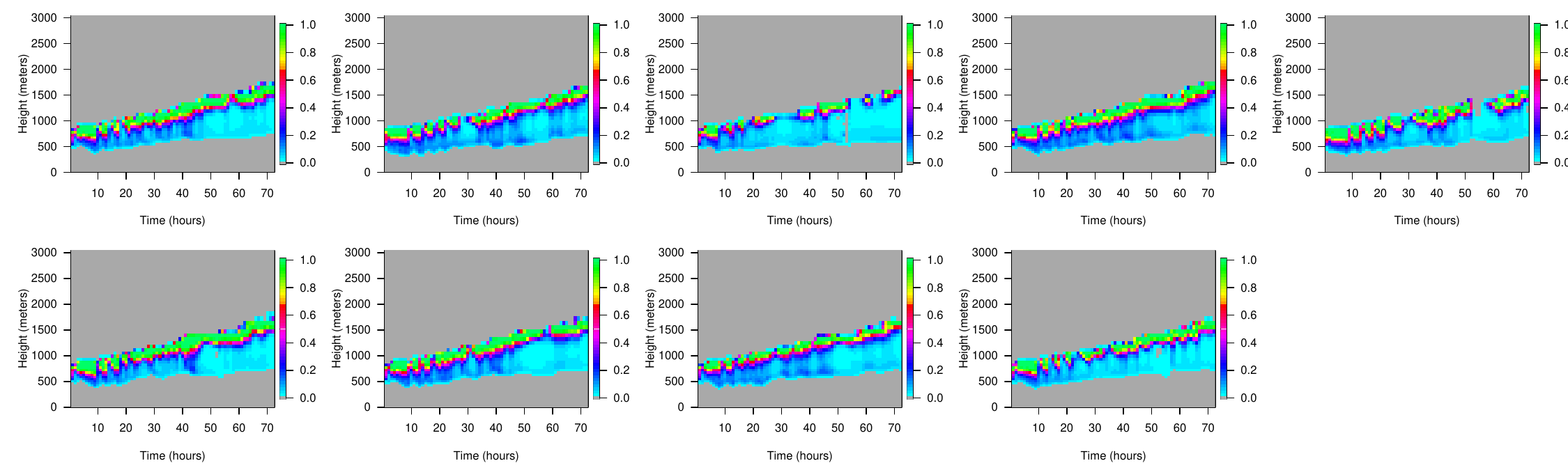}
\caption{The acceptable runs by expert's selection for wave 2.}
\label{appendix:63}
\end{figure}
\begin{figure}[H]
\centering
\includegraphics[width=14cm]{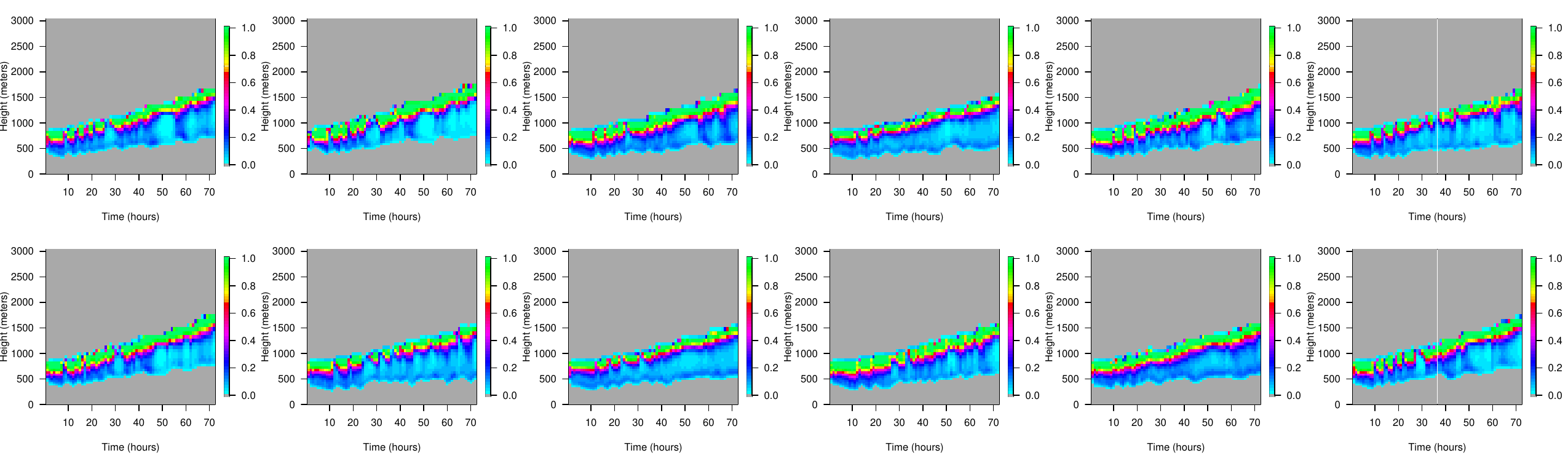}
\caption{The acceptable runs by expert's selection for wave 3.}
\label{appendix:63}
\end{figure}

\section{R shiny application}

R \texttt{Shiny} app is created for the LMDZ model calibration, aiming to consider the modeller’s information in the calibration process. In this appendix, we present the contents of the app. For interaction purpose, experts will use the app to look at the model output and accept or reject.
Hence, the \texttt{Shiny} app we created includes three pages, page 1 shows the observed field and 90 ensemble member plots, page 2  is the selection page where the experts need to choose their acceptable runs, and page 3 is used to do a final check and save the experts selection.

\begin{figure}[H]
\centering
\includegraphics[width=15cm]{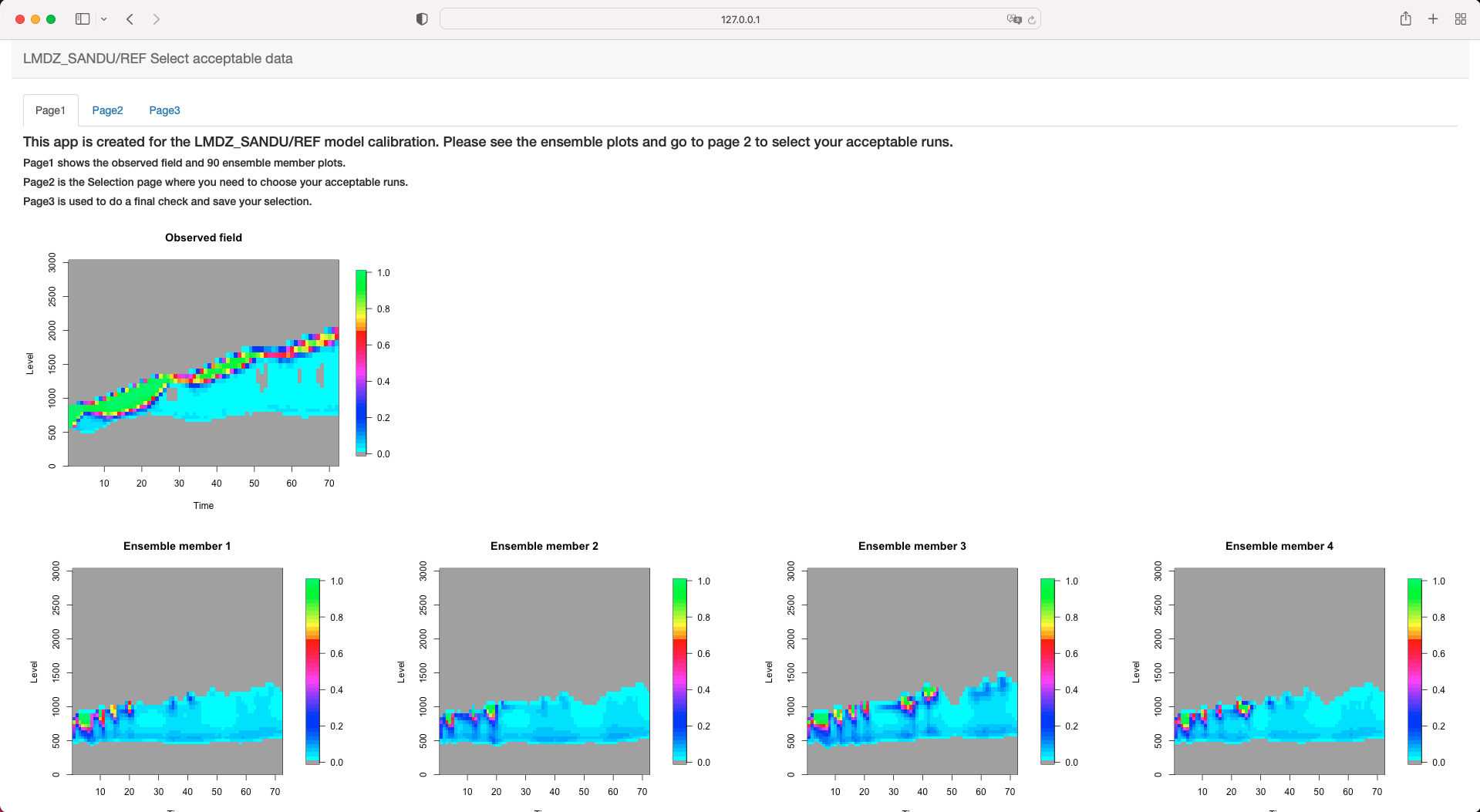}
\caption{Page 1: Overall of the ensemble.}
\label{page1}
\end{figure}

The first page of our app is presented in figure 
\ref{page1}. On the top of this page, there are three buttons with page numbers that can be used to switch the pages. Under the brief overview, the observation is fist presented, and the 90 model runs are plotted (the full page 1 is too long, we only paste part of it here, the rest can be seen by using the app). By looking at all ensemble members on page 1, experts can get an idea of which runs look best before moving to the accept reject page. 

\begin{figure}[H]
\centering
\includegraphics[width=15cm]{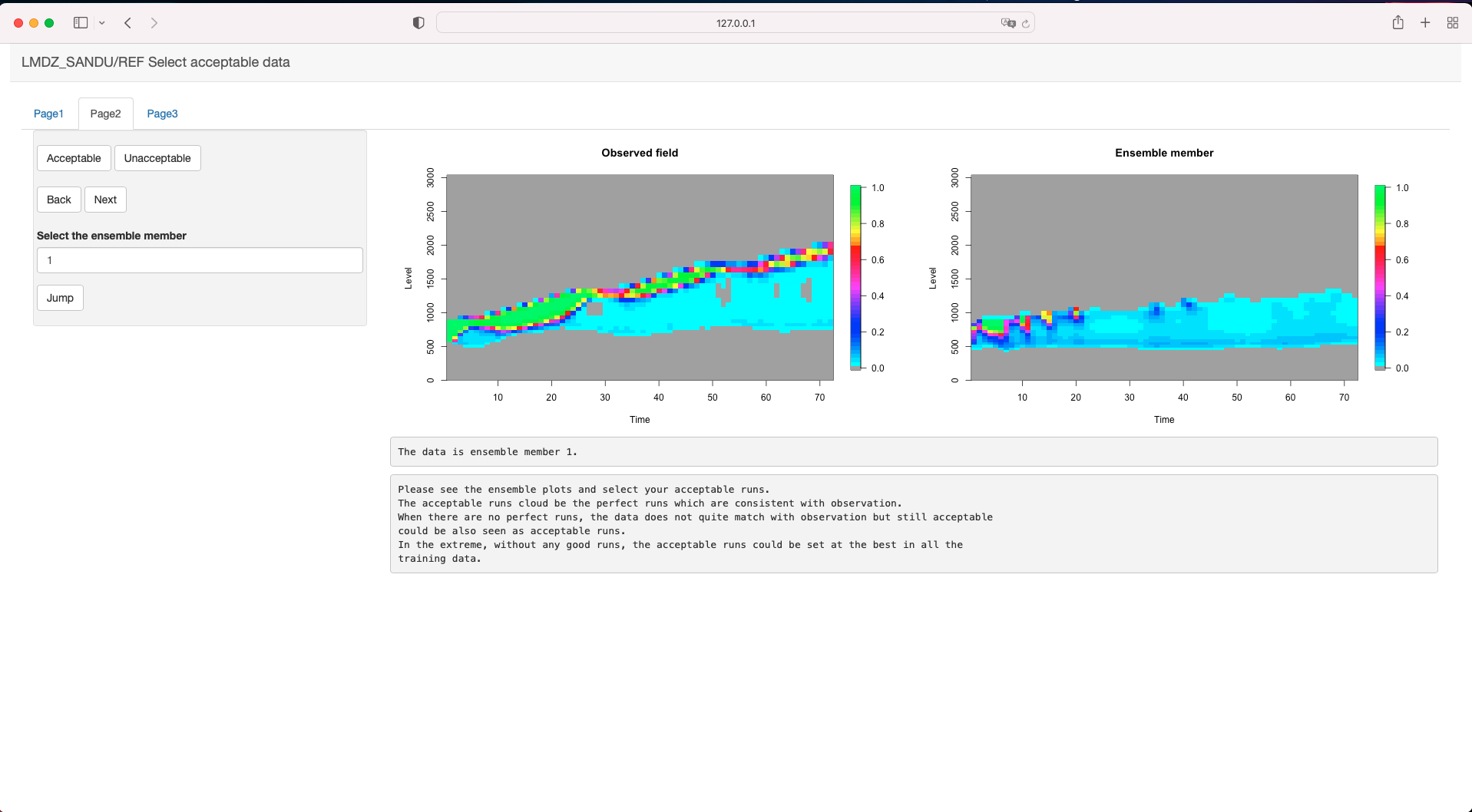}
\caption{Page 2: Selection page. }
\label{page2}
\end{figure}
The page 2 is presented in figure \ref{page2}. 
On the right panel, the left figure shows the observation, and the right figure shows the ensemble member (from the first one to ninetieth). Once the experts click the acceptable/unacceptable button on the left panel, then ensemble member will change to the next one. The ``Jump'', ``Back'' and ``Next'' buttons could be used when the experts want to correct their decisions. The app will only save their final decision. 

\begin{figure}[H]
\centering
\includegraphics[width=15cm]{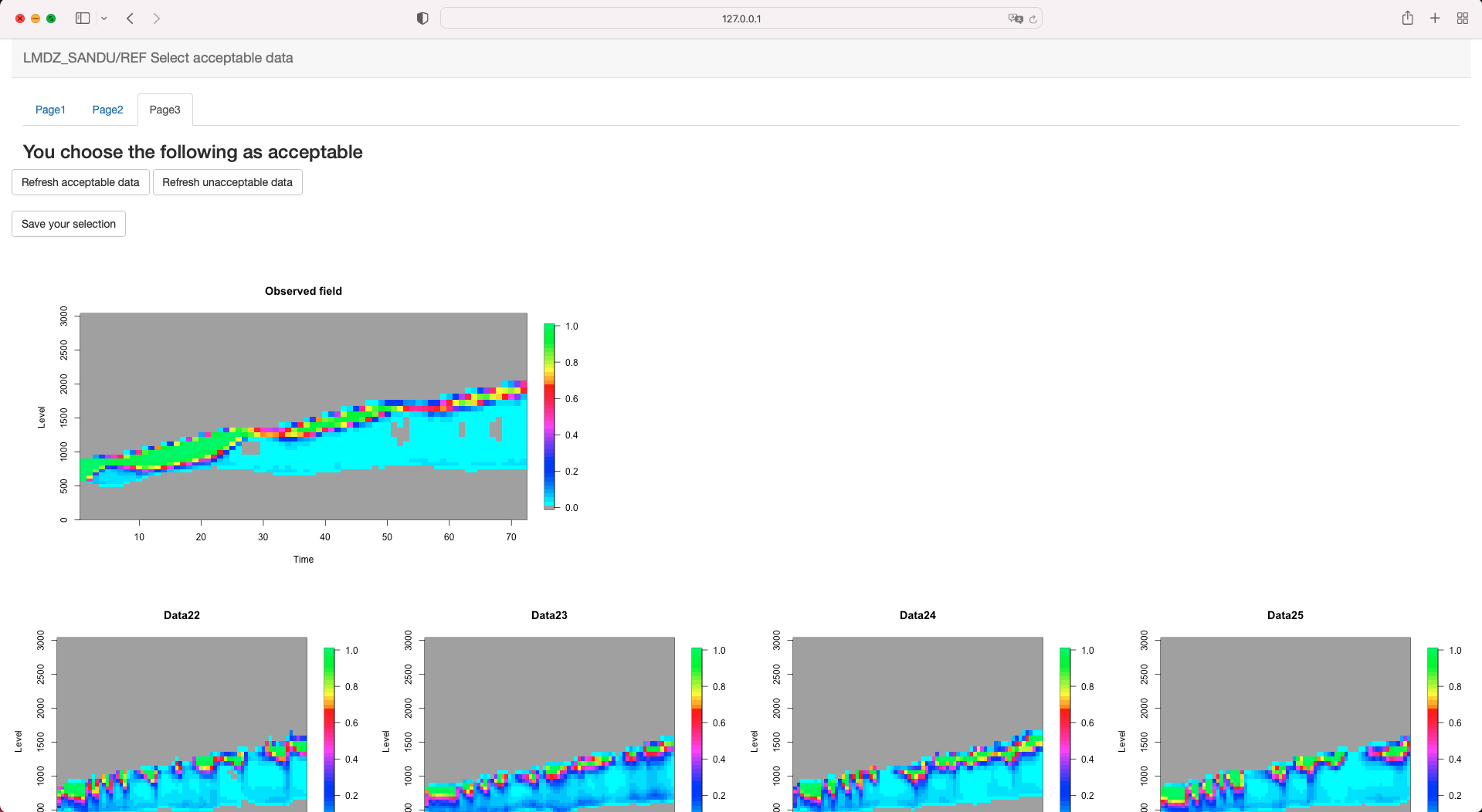}
\caption{Page 3: Final check and save the data.}
\label{page3}
\end{figure}

The page 3 is presented in figure \ref{page3} to show all of the experts selections that made in page 2. The selection will be saved by clicking the ‘save your selection’ button.  If the experts are not sure about their selection, they can go back to page 2 and type the unsure ensemble member, then they can easily compare the observed field with this ensemble member and correct the choice.

After all the steps introduced above, a csv file called ‘Acceptable.csv’ will be automatically generated. This csv file will save expert's acceptable runs as 1, expert's unacceptable runs as 2 and 0 means experts did not make any selection for this ensemble member.

\bibliographystyle{abbrvnat}
\bibliography{references}
\end{document}